\documentclass[11pt]{article}
\usepackage{amsmath,amsfonts,amssymb}
\usepackage{fullpage}
\usepackage{xspace}
\usepackage{color-edits}
\usepackage{float}
\usepackage[section]{definitions}
\usepackage[noend]{algorithmic}
\usepackage{algorithm}
\usepackage{bm}

\usepackage{hyperref}
\hypersetup{
  colorlinks=true,
  linkcolor=blue,
  citecolor=red
}

\addauthor[David]{dk}{blue}
\addauthor[Omer]{ot}{green}
\addauthor[Leonard]{ls}{red}

\newcommand{\VAL}{\ensuremath{\alpha}\xspace}
\newcommand{\Value}[1]{\ensuremath{\VAL_{#1}}\xspace}
\newcommand{\FreqVec}{\ensuremath{\vc{p}}\xspace}
\newcommand{\Freq}[1]{\ensuremath{p_{#1}}\xspace}
\newcommand{\FFR}{\ensuremath{p}\xspace}
\newcommand{\TFr}[1][]{\ensuremath{\ifthenelse{\equal{#1}{}}{p}{p_{#1}}}\xspace}
\newcommand{\EL}[1][]{\ensuremath{\ifthenelse{\equal{#1}{}}{T}{T_{#1}}}\xspace}

\newcommand{\GR}{\ensuremath{\varphi}\xspace} 

\newcommand{\SEQ}{\ensuremath{s}\xspace}
\newcommand{\Seq}[1]{\ensuremath{\SEQ_{#1}}\xspace}
\newcommand{\SEQP}{\ensuremath{s'}\xspace}
\newcommand{\SeqP}[1]{\ensuremath{\SEQP_{#1}}\xspace}
\newcommand{\SEQDIST}{\ensuremath{\Sigma}\xspace}
\newcommand{\LOC}{\ensuremath{\ell}\xspace}
\newcommand{\Loc}[1]{\ensuremath{\LOC(#1)}\xspace}
\newcommand{\LOCDIST}[1][]{\ensuremath{
\ifthenelse{\equal{#1}{}}{\Lambda}{\Lambda_{#1}}}\xspace}
\newcommand{\CLOC}{\ensuremath{\bar{\ell}}\xspace}

\newcommand{\STRAT}{\ensuremath{\mathcal{L}}\xspace}

\newcommand{\UTIL}{U}
\newcommand{\AUtilFull}[2]{\ensuremath{\UTIL(#1,#2)}\xspace}
\newcommand{\AUtilTwo}[3]{\ensuremath{\UTIL(#1,(#2,#3))}\xspace}
\newcommand{\AUtil}[1]{\ensuremath{\UTIL(#1)}\xspace}

\newcommand{\SHIFTOP}[1]{\ensuremath{M_{#1}}\xspace}
\newcommand{\ShiftOp}[2]{\ensuremath{\SHIFTOP{#1}(#2)}\xspace}
\newcommand{\SHIFTOPT}[1]{\ensuremath{\tilde M_{#1}}\xspace}
\newcommand{\ShiftOpT}[2]{\ensuremath{\SHIFTOPT{#1}(#2)}\xspace}
\newcommand{\TBV}[1][]{\ensuremath{\ifthenelse{\equal{#1}{}}{B}{B_{#1}}}\xspace} 
\newcommand{\RTIME}[1][]{\ensuremath{\ifthenelse{\equal{#1}{}}{R}{R_{#1}}}\xspace} 

\newcommand{\CDF}[1][]{\ensuremath{\ifthenelse{\equal{#1}{}}{F}{F_{#1}}}\xspace}
\newcommand{\PDF}[1][]{\ensuremath{\ifthenelse{\equal{#1}{}}{f}{f_{#1}}}\xspace}
\newcommand{\Cdf}[2][]{\ensuremath{\ifthenelse{\equal{#1}{}}{\CDF(#2)}{\CDF_{#1}(#2)}}\xspace}
\newcommand{\Pdf}[2][]{\ensuremath{\ifthenelse{\equal{#1}{}}{\PDF(#2)}{\PDF_{#1}(#2)}}\xspace}
\newcommand{\CDFH}[1][]{\ensuremath{\ifthenelse{\equal{#1}{}}{\hat{\CDF}}{\hat{\CDF}_{#1}}}\xspace}
\newcommand{\PDFH}[1][]{\ensuremath{\ifthenelse{\equal{#1}{}}{\hat{\PDF}}{\hat{\PDF}_{#1}}}\xspace}

\newcommand{\MUL}[1][]{\ensuremath{\ifthenelse{\equal{#1}{}}{\eta}{\eta_{#1}}}\xspace}

\newcommand{\OFFSET}{\ensuremath{\lambda}\xspace}
\newcommand{\CIRC}{h}
\newcommand{\circmap}[1]{\ensuremath{\CIRC(#1)}\xspace}
\newcommand{\circinv}[1]{\ensuremath{\CIRC^{-1}(#1)}\xspace}

\newcommand{\FIB}[1]{\ensuremath{f_{#1}}\xspace}

\newcommand{\Neigh}[1]{\ensuremath{\Gamma(#1)}\xspace}
\newcommand{\Slot}{\ensuremath{t}\xspace}
\newcommand{\ALLSLOTS}{\ensuremath{Z}\xspace}
\newcommand{\SLOTS}[1][]{\ensuremath{\ifthenelse{\equal{#1}{}}{W}{W_{#1}}}\xspace}
\newcommand{\SLOTSP}[1][]{\ensuremath{\ifthenelse{\equal{#1}{}}{W'}{W'_{#1}}}\xspace}
\newcommand{\Pos}[2]{\ensuremath{y_{#1,#2}}\xspace}
\newcommand{\POS}{\ensuremath{Y}\xspace}

\newcommand{\Attack}[3]{\ensuremath{(#1,#2,#3)}\xspace}
\newcommand{\AStart}{\ensuremath{t_0}\xspace}
\newcommand{\ALength}{\ensuremath{t}\xspace}

\def\dd{\mathrm{d}}

\newcommand{\e}{\ensuremath{\mathrm{e}}\xspace}
\newcommand{\lcm}{\ensuremath{\mathop{lcm}}\xspace}
\newcommand{\vc}[1]{\ensuremath{\bm{#1}}\xspace}

\newcommand{\suppress}[1]{}

\begin{document}
\title{Quasi-regular sequences and optimal schedules for security games}
\author{David Kempe\thanks{Department of Computer Science,
    University of Southern California.
    Supported in parts by NSF grants 1619458 and 1423618.}
\and Leonard J.\ Schulman\thanks{Engineering and Applied Science,
  California Institute of Technology. 
  Supported in part by NSF grants 1319745, 1618795, and a
  EURIAS Senior Fellowship co-funded by the Marie Sk{\l}odowska-Curie Actions under the 7th Framework Programme.}
\and Omer Tamuz\thanks{Departments of Economics and Mathematics,
  California Institute of Technology}}

\begin{titlepage}
\maketitle
We study \emph{security games} in which a defender commits to a mixed strategy for
protecting a finite set of targets of different values.
An attacker, knowing the defender's strategy, chooses which target to
attack and for how long.
If the attacker spends time \ALength at a target $i$ of value \Value{i},
and if he leaves before the defender visits the target,
his utility is $\ALength \cdot \alpha_i $;
if the defender visits before he leaves, his utility is $0$.
The defender's goal is to minimize the attacker's utility.
The defender's strategy consists of a schedule for visiting the targets;
it takes her unit time to switch between targets.
Such games are a simplified model of a number of real-world scenarios
such as protecting computer networks from intruders,
crops from thieves, etc.

We show that optimal defender play for such security games,
although played in continuous time,
reduces to the solution of a combinatorial question regarding
the existence of infinite sequences over a finite alphabet,
with the following properties for each symbol $i$:
(1) $i$ constitutes a prescribed limiting fraction \Freq{i} of the sequence.
(2) The occurrences of $i$ are spread apart close to evenly,
in that the ratio of the longest to shortest interval between
consecutive occurrences is bounded by a parameter $K$.
We call such sequences $K$-quasi-regular;
a $1$-quasi-regular sequence is one in which 
the occurrences of each symbol form an arithmetic sequence.

As we show, a $1$-quasi-regular sequence ensures an optimal defender
strategy for these security games: the intuition for this fact lies in
the famous ``inspection paradox.''
However, as we  demonstrate, for $K < 2$ and general \Freq{i},
$K$-quasi-regular sequences may not exist.
Fortunately, this does not turn out to be an obstruction:
we show that, surprisingly, $2$-quasi-regular sequences also suffice
for optimal defender play.
What is more, even randomized $2$-quasi-regular sequences suffice for optimality.
We show that such sequences always exist, and can be calculated efficiently.
Thus, we can ensure optimal defender play for these security games.

The question of the least $K$ for which deterministic
$K$-quasi-regular sequences exist is fascinating.
Using an ergodic theoretical approach,
we proceed to show that deterministic $3$-quasi-regular
sequences always exist (and can be calculated efficiently).
We also show that these deterministic $3$-regular sequences
give rise to a $\approx 1.006$-approximation algorithm for the
defender's optimal strategy.
For $2 \leq K < 3$ we do not know whether deterministic
$K$-quasi-regular sequences always exist;
however, when the \TFr[i] are all small, improved bounds are possible,
and in fact, $(1+\epsilon)$-quasi-regular deterministic sequences exist
for any $\epsilon > 0$ for sufficiently small \TFr[i].



\end{titlepage}

\section{Introduction} 
\label{sec:introduction} 

One of the most successful real-world applications at the
intersection of game theory and computer science has been
\emph{security games}.
Security games have been used recently to model and address
problems including the protection of infrastructure (airports,
seaports, flights), deterrence of fare evasion and smuggling,
as well as the protection of wildlife and plants.
The related model of \emph{inspection games} \cite{avenhaus2002inspection}
has been used to model interactions as varied as arms control,
accounting and auditing, environmental controls, or data verification.

In general models of security games,
there are $n$ targets of various values that the defender tries
to protect with her limited resources.\footnote{For consistency, we
  always refer to the attacker with male pronouns
  and the defender with female pronouns.} 
Different assumptions and scenarios can lead to
different interesting combinatorial constraints;
see \cite{tambe2011security} for an overview of much recent work.

In the present work, we are concerned with a defender who cannot switch
instantaneously between different targets,
which introduces a timing component and a scheduling problem.
At a high level, such constraints arise in many natural security settings, including:
\begin{enumerate}
\item Protection of computer networks (with multiple databases or
  account holders) from infiltrators. 
\item Protection of wildlife from poachers
  (e.g., \cite{fang2015green,paws2017}), crops or other plants from
  thieves, or homes in a neighborhood from burglars.
\end{enumerate}

Stripping away details, we propose the following simplified
  model for these types of settings:
If the attacker has access to an unprotected target, he gains utility in
proportion to the value of the target and to the time he spends at the
target.\footnote{In the case of access to computer systems, this
  models a scenario observed in recent attacks where the attacker
  lurks --- whether in order to monitor legitimate users, create ongoing
  damage, or because file sizes or bandwidth concerns make it
  impossible to download the entire database in a short amount of
  time.}
The game is zero-sum, i.e., the attacker's gain is the defender's loss.
If the attack is interrupted by the defender at any time,
both players receive utility 0.
Due to physical distances between targets or switching costs between databases,
the defender requires one unit of time to switch between any two targets.
The problem of interest is to determine a schedule for the defender
that will lead to minimum expected defender loss against a
best-responding attacker.

More formally, each of the $n$ targets has value $\Value{i} \geq 0$,
scaled so that $\sum_i \Value{i} = 1$.
We assume that no target is strictly more valuable than all other targets
combined,\footnote{See Section~\ref{sec:conclusion} for a discussion
  of this choice; outside of this assumption is a different r\'{e}gime
  that requires different analysis.}  so that $\Value{i}\leq 1/2$.
Time is continuous, and the game has an infinite time horizon.
A defender strategy is a schedule of which target is visited at each point in time,
with some time spent in transit.
An attacker strategy consists of a choice of a single time
interval $[\AStart,\AStart+\ALength]$ and a target $i$ to attack.
If the defender does not visit target $i$ during this interval,
the attacker obtains a utility of $\UTIL = \ALength \cdot \Value{i}$,
and the defender receives $-\UTIL$.
If the defender visits the target at any point during the interval,
both players' utilities are 0.

In our modeling of the security setting, an attack is a sufficiently
disruptive event that it effectively ends the game, by which we mean
that the defender will subsequently re-randomize her schedule.  Thus,
the entire game is concluded after one attack.  We assume that the
defender's mixed strategy (but not randomness) is known to the
attacker who will best-respond by choosing an attack
\Attack{i}{\AStart}{\ALength}, comprising the target $i$ and the start
time and duration\footnote{%
  One could consider an ``adaptive'' attacker, who initially only
  chooses the start time \AStart, and decides on \ALength on the fly.
  The resulting model would be equivalent, as such an adaptive
  attacker at any time $t' > \AStart$ has learned no new information
  aside from the fact that he has not yet been caught; any attacker
  who under these conditions will decide to wait for exactly \ALength
  units of time (if not caught) is exactly equivalent to one who
  chooses the attack \Attack{i}{\AStart}{\ALength} at once.}
of the attack, \AStart and \ALength.
Accordingly, the defender chooses a Min-Max strategy: a strategy which
minimizes the maximum expected return of any attack
\Attack{i}{\AStart}{\ALength}.

Although the game is played in continuous time,
we show (in Section~\ref{sec:preliminaries})
that under the assumption that $\Value{i} \leq 1/2$ for all $i$,
optimal defender strategies can be obtained as follows:
choose a suitable (random) \emph{discrete sequence}
$\Seq{0}, \Seq{1}, \Seq{2}, \ldots$ of targets
and a uniformly random $\tau \in [0,1]$;
then visit each target \Seq{t} at time $t+\tau \in \R^+$.
At all other times, the defender is in transit,
so that each target is visited instantaneously,
with the next target visited exactly one time unit later.

It is not difficult to show that if each target $i$ occupies a
\Value{i} fraction of the sequence and is exactly evenly spaced in the
sequence, the resulting defender schedule is optimal.
Some intuition is derived from the famous ``Inspection Paradox'' or
``Waiting Time Problem:''
passengers of a bus service which departs a station with perfect regularity
(e.g., 15 minutes apart) wait on average half as long as passengers of
a service with the same frequency of operation but Poisson departure times.
In our case, higher variance in the defender's interarrival times
lengthens the expected time until the next defender visit,
making longer attacks more attractive.

We call sequences in which each target $i$ is
exactly evenly spaced \emph{regular};
generalizing this notion, we call a sequence \SEQ $K$-quasi-regular
(with respect to frequencies \Freq{i}; in our applications,
we typically choose $\Freq{i} = \Value{i}$)
if, as before, each target $i$ takes up a \Freq{i} fraction of the sequence,
and the ratio of the longest to shortest interval between consecutive
occurrences of $i$ in \SEQ is bounded by $K$.
Our first result (Theorem~\ref{thm:optimal-distributions}) is that
--- surprisingly --- it suffices for optimality that the sequence is
\emph{$2$-quasi-regular}.

\subsection{Our Main Result}

It is fairly straightforward to show that there is some
vector $(\Freq{i})_i$ such that there are no
$(2-\epsilon)$-quasi-regular sequences for any $\epsilon > 0$;
we do this in Section~\ref{sec:optimal-schedule}.
Our main result (Theorem~\ref{thm:optimal-powers-of-2} in
Section~\ref{sec:optimal-schedule}) is that for any values
\Freq{i}, there exists a 2-quasi-regular random sequence,
which can furthermore be efficiently computed from the \Freq{i}. 
By the aforementioned Theorem~\ref{thm:optimal-distributions}, 
the corresponding defender mixed strategy is optimal.

\subsection{Ergodic Schedules}

Quasi-regular sequences are basic combinatorial objects, quite apart
from our application of them.
One limitation of our work (although it does not affect the
application) is that the resulting schedules are not ergodic:
they randomize between different schedules in which the
targets have frequencies differing from the desired \Freq{i}.
It is then a natural question whether 2-quasi-regular ergodic sequences can
be obtained as well.
This is related to the following combinatorial question:
given densities \Freq{i},
does there always exist a $2$-quasi-regular deterministic sequence?

We provide two partial answers to this question.
In Section~\ref{sec:golden-ratio},
we analyze a very simple schedule called the \emph{Golden Ratio Schedule}
(variations of which have been studied in the context of
hashing~\cite[pp.~510,511,543]{knuth1998art}, bandwidth
sharing~\cite{itai1984golden,panwar1992golden} and elsewhere). 
This schedule is generated by the following random sequence: 
partition the circle of circumference $1$ into intervals of size \Freq{i}
corresponding to the targets $i$.
Choose a uniformly random starting point on the circle. 
In each step, add \GR to the current point;
here, $\GR = \half(1+\sqrt{5})$ is the Golden Ratio.
In each time step, the defender visits the target $i$ into
whose interval the current point falls.

This random sequence is ergodic, and at worst
3-quasi-regular.\footnote{When all $\Freq{i} \leq 1-1/\GR$, the
  quasi-regularity guarantee improves to $8/3$,
  and as $\Freq{i} \to 0$, it converges to $\GR^2$.}
Moreover, for any choice of the random starting point,
the deterministic sequence is 3-quasi-regular.
Thus we show that there always exist deterministic 3-quasi-regular
sequences.
We do not know if this is true for any $K < 3$.

It is interesting that such a simple schedule achieves constant
quasi-regularity, but the bound is not strong enough to guarantee
optimality of the schedule for the defender.
However, we show that the schedule is nearly optimal for the defender:
the attacker's utility is within a factor of at most $1.006$ of the
minimum attacker utility (and this bound is tight).  
The proof of this approximation guarantee relies on a theorem of
Slater about simple dynamical systems like the Golden Ratio shift,
and a somewhat intricate analysis of the attacker's response.
We find it remarkable that such a simple policy comes provably within
0.6\% of the optimum, in particular compared to another very simple policy:
as we show in Appendix~\ref{sec:geometric},
the simple \emph{i.i.d.~schedule},
which always chooses the next target $i$ to visit with probability \Value{i},
independent of the history, is only a $4/\e$-approximation.

\smallskip

As a second partial result towards obtaining an optimal ergodic
schedule, in Section~\ref{sec:matching},
we show a sufficient condition for the existence
of $(1+\epsilon)$-quasi-regular sequences, for any $\epsilon > 0$.
Specifically, let $M$ be the smallest common denominator of all \Freq{i}.
If $\Freq{i} = O(\frac{\epsilon}{\sqrt{n \log M}})$ for all $i$,
then a $(1+\epsilon)$-quasi-regular periodic schedule exists and can
be found efficiently using a randomized algorithm that succeeds with high probability.

The algorithm is based on placing points for target $i$ at uniform
distance proportional to $1/\Freq{i}$ on the unit circle, 
with independent uniformly random offsets. 
Points can only be matched to sufficiently close multiples of $1/M$.
An application of Hall's Theorem,
similar to~\cite{tijdeman1973distribution,holroyd2010rotor},
shows that under the conditions of the theorem, 
this algorithm succeeds with high probability in producing a
$(1+\epsilon)$-quasi-regular sequence.

\subsubsection*{Related Work}

Several notions of ``sequences in which elements $i$ are roughly
regularly spaced, with given frequencies \Freq{i}'' have been studied
in different contexts.

Dating back to the work of Tijdeman
\cite{tijdeman1973distribution,tijdeman1980chairman}, several papers
\cite{angel-discrepancy,holroyd2010rotor,tijdeman1973distribution,tijdeman1980chairman,chung2016discrepancy}
(see also an overview in~\cite{chen2014panorama}) have studied
sequences with low \emph{discrepancy} in the following sense:
up to any time $t$, the number of occurrences of element 
$i$ approximates $t \cdot \Freq{i}$ as closely as possible.
For our application, the rate of convergence of the frequencies to
\Freq{i} is not essential;
but it is crucial that the defender's interarrival times at each
target be as regular as possible.
Consequently, methods from this literature are not
sufficient to optimally solve our problem.

In the \emph{Pinwheel Problem} \cite{HMRTV:pinwheel-conference,holte:rosier:tulchinsky:varval:pinwheel,chan:chin:pinwheel,lin:lin:pinwheel,fishburn:lagarias:pinwheel},
one is given integers $n_1, n_2, \ldots, n_k \geq 2$, whose density is defined by
$\beta := \sum_i 1/n_i \leq 1$.
The goal is to produce a sequence \SEQ of the elements
$\SET{1, \ldots, k}$ such that for each $i$, each subsequence
$(\Seq{t}, \Seq{t+1}, \ldots, \Seq{t+n_i-1})$ of length $n_i$
contains at least one occurrence of the symbol $i$.
One of the main questions in this area
is what values $\beta$ are sufficient to guarantee
the existence of such a sequence.
Currently, it is known that $\beta \leq 5/6$ is necessary (i.e., there
are examples with $\beta = 5/6 + \epsilon$ such that no
\SEQ exists), and $\beta \leq 7/10$ is sufficient. 
In our case, we always have desired frequencies adding up to 1,
and we have not only (looser) upper bounds,
but also lower bounds on the distance between consecutive occurrences.
While some of the basic techniques in the context of the
Pinwheel Problem are similar to the ones we use
(in Section~\ref{sec:optimal-schedule}), 
solutions for one problem do not imply solutions to the other.

In concurrent and independent work,
Immorlica and Kleinberg~\cite{immorlica-kleinberg-bandits}
--- motivated in part also by applications to preventing wildlife
poaching --- defined a ``Recharging Bandits'' problem in which the
available reward at targets grows according to (unknown) concave
functions.
In the full-information setting (in which the functions are known),
(near-)optimal solutions correspond to schedules that visit targets
with given frequencies at roughly evenly spaced intervals.
The precise definition of ``roughly evenly spaced'' differs from ours,
and while some of the techniques used in
\cite{immorlica-kleinberg-bandits} are similar to ours,
optimality results and approximation guarantees do not imply each
other in either direction.

Our work is related to the inspection games literature
(see, e.g.,~\cite{avenhaus2002inspection,vonstengel2016recursive}).
Specifically, several
articles~\cite{diamond1982minimax,avenhaus2005playing,avenhaus2013distributing}
consider models in which an inspectee indulges in illegal activity
once or multiple times within a finite time interval.
An inspector distributes optimally the times at which she performs a
fixed number of inspections, and suffers a loss that is proportional
to the time that has elapsed between the beginning of illegal
activity and the next inspection.
In these models, as in ours, the inspector wants to visit
inspectees regularly while keeping the inspectee uncertain about
visit times.
The lack of travel time restrictions as well as the lack of a need to
catch the inspectee at the time of his action make the models
mathematically incomparable.

Finally, our work is also related to the literature on
\emph{patrol games}. Here, as in our model, a defender (or multiple
cooperating defenders) must decide on a schedule of visits to
targets facing attacks. However, unlike our model, the attacker
observes the past locations of the defender(s) before deciding
whether and where to attack (see, e.g.,~\cite{basilico2009leader,
vorobeychik2012adversarial}).

\section{Preliminaries}
\label{sec:preliminaries}

The $n$ targets have values $\Value{i} > 0$ for all $i$.
Because the units in which target values are measured are irrelevant,
we assume that $\sum_i \Value{i} = 1$.
We assume that no target has value exceeding the sum of all other
targets' values, meaning (after normalization) that
$\Value{i} \leq \half$ for all $i$.

A \emph{pure strategy} (\emph{schedule})
for the defender is a measurable mapping
$\LOC: \R^{\geq 0} \to \SET{1, 2, \ldots, n, \perp}$,
where $\perp$ denotes that the defender is in transit. 
A schedule \LOC is \emph{valid} if $\Loc{t} = i$ and
$\Loc{t'} = j \neq i$ implies that $|t'-t| \geq 1$.
In other words, there is enough time for the defender to move from $i$
to $j$ (or from $j$ to $i$).
We use \STRAT to denote the set of all valid pure defender strategies.

The defender moves first and chooses a mixed strategy,
i.e., a distribution \LOCDIST over \STRAT, or a random \LOC.
Then, the attacker chooses an attack \Attack{i}{\AStart}{\ALength}
consisting of a target $i$, a start time \AStart, and an attack
duration \ALength.
Subsequently, a mapping \LOC is drawn from the defender's distribution
\LOCDIST.
The attacker's utility is
\begin{align}
\AUtilFull{\LOC}{\Attack{i}{\AStart}{\ALength}}
& = \begin{cases}
0  & \mbox{ if } \LOC(\tau) = i \mbox{ for some } \tau \in [\AStart, \AStart+\ALength]\\
\Value{i} \cdot \ALength & \mbox{ otherwise}.
\end{cases} \label{eqn:attacker-util}
\end{align}
Since we are considering a zero-sum game
(see Section~\ref{sec:conclusion} for a discussion),
the defender's utility is
$-\AUtilFull{\LOC}{\Attack{i}{\AStart}{\ALength}}$.
Note that the attacker attacks only once.

A rational attacker will choose \Attack{i}{\AStart}{\ALength} 
so as to maximize
$\Expect[\LOC \sim \LOCDIST]{\AUtilFull{\LOC}{\Attack{i}{\AStart}{\ALength}}}$;
therefore, the defender's goal is to choose \LOCDIST to minimize
\begin{align*}
  \AUtil{\LOCDIST} =   \sup_{i,\AStart,\ALength}\Expect[\LOC \sim
  \LOCDIST]{\AUtilFull{\LOC}{\Attack{i}{\AStart}{\ALength}}}.
\end{align*}

The fact that we assume an infinite time horizon is primarily an
idealization, in order to avoid mathematical inconveniences at the end
of the time horizon. The reader is encouraged to think of the
``infinite'' time horizon as one or a few days, long enough that a
significant schedule needs to be planned and boundary effects can be
ignored at small cost; but short enough that the attacker cannot
observe early parts of the schedule to infer which schedule \LOC was
drawn from \LOCDIST.

\subsection{Canonical, Shift-Invariant, and Ergodic Schedules}
The general definition of defender schedules allows for strange
schedules that are clearly suboptimal.
We would like to restrict our attention to ``reasonable'' schedules.
In particular, we will assume the two following conditions, which we
later show to hold without loss of generality.
(Here, we will be slightly informal in our definitions.
Precise definitions and constructions ensuring these properties are
given in Appendix~\ref{sec:appendix-formalization}.)

\begin{itemize}
\item Whenever the defender transitions from one target $i$ to a
  target $i'$ ($i' = i$ is possible), she spends exactly one time unit
  in transit.  We call such schedules \emph{canonical}.
\item To the attacker, any two times $t$ and $t'$ ``look the same,''
  in that for any $t,t',\tau \in \R^+$, the distributions of the
  defender's schedule restricted to the time intervals $[t,t+\tau]$
  and $[t',t'+\tau]$ are the same.
  We call such schedules \emph{shift-invariant} or \emph{stationary}.
\end{itemize}

Because the strategy spaces of both players are infinite,
it is not clear a priori that a Min-Max schedule for the defender exists.
In Appendix~\ref{sec:appendix-formalization},
we show that an optimal mixed Min-Max defender strategy does
exist, and is w.l.o.g.~canonical and shift-invariant.
Therefore, for the remainder of this paper,
we will focus only on shift-invariant canonical schedules.
When the defender's strategy is shift-invariant,
the start time \AStart of the attack does not matter,
so shift-invariance allows us to implicitly assume that the attacker
always attacks at time $\AStart = 0$.
We then simply write \AUtilTwo{\LOCDIST}{i}{\ALength} for the attacker's
expected utility from attacking target $i$ for \ALength units of time.

One may additionally be interested in constructing
\emph{ergodic} shift-invariant mixed schedules:
\LOCDIST is ergodic if \LOCDIST cannot be
written as the convex combination
$\LOCDIST = \lambda \LOCDIST_1 + (1-\lambda)\LOCDIST_2$
of two different shift-invariant mixed schedules
(see the formal discussion in Appendix~\ref{sec:appendix-formalization}).
While we are not aware of game-theoretical implications of ergodicity,
the question is mathematically natural,
and may be important in some extensions of our model.

\subsection{Return Times and Target Visit Frequencies}
  
For the following definitions, recall that we are 
focusing only on shift-invariant canonical schedules.
Accordingly, we assume without loss of generality
that the attacker starts his attack at time 0.
Since the attacker chooses only the target $i$ and the duration $t$ of
the attack, from his perspective, the property of the defender's
strategy that matters is the distribution of her
next \emph{return time} to target $i$, defined as
$\RTIME[i] = \min \Set{t \geq 0}{\Loc{t} = i}$.
Given a target $i$ and a defender strategy, let
$\Cdf[i]{t} = \Prob{\RTIME[i] \leq t}$
denote the CDF of \RTIME[i].
In particular, notice that \Cdf[i]{0} is the fraction of time the
defender spends waiting at target $i$.
In terms of the distribution of return times \CDF[i], the attacker's
utility can be expressed as follows:

\begin{align}
\AUtilTwo{\CDF[i]}{i}{t} & = \Value{i} \cdot t \cdot (1-\Cdf[i]{t}).
\label{eqn:expected-attacker-util}
\end{align}

Next, we define a random variable \TBV[i] capturing the (random)
duration between consecutive visits to the same target $i$.
Notice the subtle difference between this quantity and the time from
the defender's perspective between leaving a target and returning to
it.
By comparison, the distribution of \TBV[i] should assign
higher probability to higher values:
as in the inspection paradox, at a random point in time,
the attacker is more likely to find himself in a large gap.
In other words, larger gaps are more likely to appear at a fixed time
than on average over a long time stretch.

Defining the distribution for \TBV[i] precisely requires some care,
and is done formally in Appendix~\ref{sec:appendix-formalization}.
The construction formalizes the following intuition:
we can consider the limit as we shift the random schedule \LOC
left by a real number $t \to \infty$. 
This extends a shift-invariant random schedule from the
non-negative reals to all reals,
and allows us to consider the random variable
\[
  \TBV[i] = (\inf \Set{t \in \R, t > 0}{\Loc{t} = i}) 
         - (\sup \Set{t \in \R, t \leq 0}{\Loc{t} = i}).
\]
The random variable \TBV[i] captures the time between consecutive
visits before and after time 0 ---
by shift-invariance, the time 0 is arbitrary here.
Our constructions will always ensure that \TBV[i] is finite.






We can write \CDF[i] in terms of \TBV[i].
First, we note that conditioned on $\TBV[i] = \tau$,
the distribution of \RTIME[i] is uniform between 0 and $\tau$.
Hence, its conditional CDF is
$\ProbC{\RTIME[i] \leq t}{\TBV[i]=\tau} = \min(1,t/\tau)$.
The unconditional CDF of \RTIME[i] for $t > 0$ is then
\begin{align}
  \label{eq:cdf-tbv}
  \Cdf[i]{t} 
 & = \Expect[{\TBV[i]}]{\min(1,t/\TBV[i])} 
 \; = \; \int_0^1 \Prob{t/\TBV[i] \geq \tau} \dd\tau
 \; = \; \int_0^1 \Prob{\TBV[i] \leq t/\tau} \dd\tau
 \; = \; t\cdot\int_t^\infty\frac{\Prob{\TBV[i] \leq \tau}}{\tau^2}\,\dd\tau;
\end{align}
the final equality uses a change of integration variable.
For $t=0$, we have that $\Cdf[i]{0} = \Prob{\TBV[i]=0}$ is
the probability that the defender is at target $i$ at time $0$.
Two key quantities for our analysis are
\begin{align*} 
  \TFr[i] &= \Cdf[i]{1}, &  \EL[i] &= \frac{1-\Cdf[i]{0}}{\TFr[i]-\Cdf[i]{0}}.
\end{align*}
In every canonical schedule, \TFr[i] is the fraction of time that the
defender is either at target $i$ or in transit to $i$.
Thus, $\sum_i\TFr[i]=1$ in every canonical schedule.
\EL[i] intuitively captures the ``expected time'' between
consecutive visits to target $i$, as seen by the defender.
However, this intuition formally holds only for periodic schedules,
necessitating the preceding more complex definition for arbitrary
schedules.
The most useful facts about \CDF[i] are summarized by the following
proposition:
\begin{proposition} \label{lem:F-TFr-generalized}
\begin{enumerate}
  \item $\Cdf[i]{\cdot}$ is concave.
  \item $\Cdf[i]{t} \leq \Cdf[i]{0} + (\TFr[i] - \Cdf[i]{0}) \cdot t$,
    with equality iff $\Prob{\TBV[i]=0 \mbox{ or } \TBV[i] > t} = 1$.
\end{enumerate}
\end{proposition}

\begin{emptyproof}
  \begin{enumerate}
  \item For each $x$, the function $\min(1,t/x)$ is concave in $t$.
    Hence, $\Cdf[i]{t} = \Expect[x]{\min(1,t/x)}$, being a
    convex combination of concave functions, is concave.
  \item
    Rearranging and applying~\eqref{eq:cdf-tbv}, we want to show that
    \begin{align*}
      t \cdot \int_1^\infty\frac{\Prob{\TBV[i] \leq \tau}}{\tau^2}\,\dd\tau 
      - t \cdot \int_t^\infty\frac{\Prob{\TBV[i] \leq \tau}}{\tau^2}\,\dd\tau
     & \geq (t-1) \cdot \Cdf[i]{0}.
    \end{align*}
    Combining the two integrals and lower-bounding
    $\Prob{\TBV[i] \leq \tau} \geq \Cdf[i]{0}$ yields that the left hand
    side is at least
    \begin{align*}
      \Cdf[i]{0} \cdot t \cdot \int_1^t\frac{1}{\tau^2}\,\dd\tau 
     & = \Cdf[i]{0} \cdot (t-1),
    \end{align*}
    with the lower bound holding with equality iff
    $\Prob{\TBV[i]=0 \mbox{ or } \TBV[i] > t} = 1$.\QED
  \end{enumerate}
\end{emptyproof}

\subsection{Schedules from Sequences}
All the constructions of mixed defender schedules in this paper will
have the property that the defender never waits at any target,
instead traveling immediately to the next target.
That such schedules are optimal
(and hence the focus on such constructions is w.l.o.g.)
is a consequence of our main Theorem~\ref{thm:optimal-distributions},
and hinges on the restriction that $\Value{i}\leq 1/2$ for all $i$.
A brief discussion of what happens when this assumption is relaxed is
given in Section~\ref{sec:conclusion}.

Canonical schedules without waiting are readily identified with
schedules defined only on integer times, since the defender must only
choose, after visiting a target, which target she will visit next.
We call such schedules \emph{sequences},
defined as $\SEQ \colon \N \to \SET{1,\ldots,n}$.
A sequence, together with a start time $t_0$,
naturally defines a canonical schedule,
by setting $\Loc{t} = \Seq{t-t_0}$ if $t-t_0 \in \N$,
and $\Loc{t} = \perp$ otherwise.
\SEQDIST denotes a distribution over sequences,
or the distribution of a random sequence \SEQ.

Shift-invariance (or stationarity) can be defined for random
sequences as for (continuous) mixed schedules.
When \SEQ is a periodic sequence, i.e., there is a $k$ such that
$\Seq{t+k} = \Seq{t}$ for all $t$, a shift-invariant random sequence
can be obtained particularly easily, by choosing a uniformly random
$\kappa \in \SET{0, \ldots, k-1}$, and defining $\SEQP$ via
$\SeqP{t} = \Seq{t+\kappa}$; for aperiodic sequences, we give
  a construction in Appendix~\ref{sec:appendix-formalization}.  From a
shift-invariant random sequence, we can obtain a shift-invariant mixed
schedule straightforwardly, by choosing the start time $t_0 \in [0,1]$
uniformly.

For the special case of random sequences,
the definitions of \TFr[i] and \EL[i] simplify to
$\TFr[i] = \Prob{\Seq{1} = i}$ (which is now exactly the fraction of
target visits devoted to target $i$),
and $\EL[i] = 1/\TFr[i]$ (since $\Cdf[i]{0} = 0$).

\subsection{Regular and Quasi-Regular Sequences}

We say that a shift-invariant random sequence \SEQ is
$K$-quasi-regular (with respect to target frequencies \Freq{i})
if the following two hold for each target $i$:
\begin{enumerate}
\item $\Prob{\Seq{1}=i} = \Freq{i}$.
\item There is some $b_i$ such that 
$\Prob{b_i \leq \TBV_i \leq K \cdot b_i}=1$.
\end{enumerate}
In other words, each target $i$ is visited with frequency \Freq{i},
and the maximum gap for consecutive visits to target $i$ is
within a factor $K$ of the minimum gap with probability 1.
A random sequence is \emph{regular} if it is 1-quasi-regular,
meaning that all visits to target $i$ are spaced exactly \EL[i] apart.
(All definitions extend directly to canonical, mixed, shift-invariant
schedules.)

A particularly straightforward way to obtain a $K$-quasi-regular
random sequence \SEQDIST is to consider the a subsequential
limit of uniformly random shifts of a deterministic sequence
\SEQ in which the gaps between consecutive visits to $i$ are bounded
between $b_i$ and $K b_i$, and the density of entries which are $i$ is
\Freq{i}.

\section{The Attacker's Response, and Optimal Schedules} \label{sec:utility-to-schedule}

In this section, we show the following main theorem, a sufficient
condition for a random sequence to be optimal for the defender.



\begin{theorem}
  \label{thm:optimal-distributions}
  Consider a random shift-invariant sequence such that the following
  two hold for each target $i$:
\begin{itemize}
\item $\EL[i] = 1/\Value{i}$.
\item For each $i$, there exists an \MUL[i] such that
  $\Prob{\frac{\MUL[i]}{\MUL[i]+1}\EL[i] \leq \TBV_i \leq \MUL[i] \EL[i] }=1$. 
\end{itemize}
Then, the associated mixed strategy is optimal for the defender. 

In particular, these conditions hold for $2$-quasi-regular random
sequences with respect to \Value{i}.
\end{theorem}

In Section~\ref{sec:optimal-schedule}, we show that there always exists
a $2$-quasi-regular sequence. 
With the eventual goal of proving Theorem~\ref{thm:optimal-distributions},
we fix a target $i$, and for now drop the subscript $i$, so that
\begin{align*}
  \TFr = \TFr[i]  & &\Cdf{t} = \Cdf[i]{t} &
  &\EL = \EL[i] && \TBV = \TBV[i].
\end{align*}
We fix  \TFr and \EL and study which sequences --- among all those
with given \TFr and \EL --- are optimal for the defender.
Our algorithmic constructions will choose
$\TFr[i] = \Value{i}$ for all $i$;
however, in order to show the optimality of this choice,
the following proposition and corollary are proved for general
$\TFr, \EL$.
  
\begin{proposition} \label{lem:attacker-minimum}
Consider any canonical shift-invariant mixed defender schedule (over
the non-negative real numbers).
By choosing $t=\EL/2$, 
the attacker guarantees himself a utility of at least
$\VAL \cdot \frac{1-\Cdf{0}}{4} \cdot \EL$.
\end{proposition}
\begin{proof}
By Equation~\eqref{eqn:expected-attacker-util}, 
the attacker's utility at time $t=\EL / 2$ is 
$\VAL\cdot (\EL/2) \cdot (1-\Cdf{\EL/2})$.
Using Proposition~\ref{lem:F-TFr-generalized}, 
we can bound
\[
1-\Cdf{\EL/2}
\; \geq \;
1-\Cdf{0} - (\TFr-\Cdf{0}) \cdot (\EL/2)
\; = \; 
(1-\Cdf{0}) \cdot \left(1- \frac{\TFr-\Cdf{0}}{1-\Cdf{0}} \cdot
(\EL/2) \right)
\; = \; \frac{1-\Cdf{0}}{2}.
\]
Hence, the attacker's utility is at least
$\VAL\cdot\frac{1-\Cdf{0}}{4} \cdot \EL$.
\end{proof}


Recalling that a random \emph{sequence} by definition
does not involve waiting at any target,
we obtain the following simple corollary about random
sequences that are worst for the attacker:
\begin{corollary} \label{lem:sufficient-optimal}
Among random sequences with fixed \EL and \TFr,
any random sequence is optimal for the defender if the attacker's
payoff is upper-bounded by $\quarter \cdot \VAL \cdot \EL$.
\end{corollary}

The following corollary is particularly useful:
\begin{corollary} \label{lem:attacker-zero}
Fix \EL and \TFr, and consider a random sequence in which the
attacker's optimal attack duration $t$ satisfies $\Prob{\TBV > t}=1$.
Then, this random sequence is optimal for the defender. 
Furthermore, in this case, w.l.o.g., $t = \EL/2$.
\end{corollary}

\begin{proof}
By the assumption that $\Prob{\TBV > t}=1$
and Proposition~\ref{lem:F-TFr-generalized}, we have that
$\Cdf{t} = \TFr \cdot t$.
Hence, the attacker's utility is 
$\VAL \cdot (1 - \TFr \cdot t) \cdot t
= 
\VAL \cdot \frac{t \cdot (\EL-t)}{\EL}
\leq 
\VAL\cdot \frac{\EL}{4}.$
  Now, the claim follows directly from 
  Corollary~\ref{lem:sufficient-optimal}.
  That $t=\EL/2$ is a best response follows from
  Proposition~\ref{lem:attacker-minimum}.
\end{proof}

We can now apply these corollaries to show optimality 
for a single target for which the ``quasi-regularity'' of return times
holds. 

\begin{proposition} \label{lem:optimal-distributions}
Fix \EL and \TFr, and consider a random sequence such that for some \MUL,
\begin{align*}
\Prob{\frac{\MUL}{\MUL+1}\EL \leq \TBV \leq \MUL \EL} = 1.
\end{align*}
Then, this schedule is optimal for the defender among schedules with
these \EL and \TFr.
\end{proposition}

\begin{proof}
We write $\xi = \frac{\MUL}{\MUL+1}$.
By Proposition~\ref{lem:attacker-minimum}, choosing $t = \EL/2$, 
the attacker can guarantee himself at least a utility of 
$\quarter \cdot \VAL \cdot \EL$.
We will show below that the attacker's utility for any attack duration
$t \in [\xi \EL, \infty)$ is at most
$\quarter \cdot \VAL \cdot \EL$.

Hence, the attacker has an optimal attack duration
$t \leq \xi \EL$ (either $t = \EL/2$ or a different $t$).
By the assumption that $\Prob{\xi \EL \leq \TBV} = 1$
and Proposition~\ref{lem:F-TFr-generalized},
$\Cdf{\xi \EL} = \TFr \cdot \xi \EL$.
Using the concavity of \CDF, this implies that
$\Cdf{t} = \TFr \cdot t$ for all $t \leq \xi \EL$.
Thus, whichever such $t$ is optimal for the attacker,
Corollary~\ref{lem:attacker-zero} implies that \CDF
is optimal for the defender, and furthermore,
that $t = \EL/2$ is optimal for the attacker after all.

It remains to prove the upper bound for $t \geq \xi \EL$.
For any $t \geq \MUL \EL$, the assumption 
that $\TBV \leq \MUL\EL$ with probability 1 implies that
$\Cdf{\MUL \EL} = 1$, and hence a utility of 0 for the attacker. 
So we focus on $t \in [\xi \EL, \MUL \EL]$, and show that in
this range, the maximum utility of the attacker is 
at most $\VAL \cdot \EL/4$. 

We proved above that $\Cdf{\xi \EL} = \xi$,
and by the assumption $\Prob{\TBV \leq \MUL \EL} = 1$,
we get that $\Cdf{\MUL \EL} = 1$. 
Since \CDF is concave by Proposition~\ref{lem:F-TFr-generalized},
for $t \in [\xi \EL, \MUL \EL]$,
\CDF is bounded below by the line connecting
$(\xi \EL, \xi)$ and $(\MUL \EL, 1)$, so
\begin{align*}
\Cdf{t} & \geq 
\xi + \frac{t-\xi \EL}{(\MUL-\xi) \EL} \cdot (1- \xi)
\; = \; \xi \cdot \frac{\MUL-1}{\MUL-\xi} 
      + \frac{1-\xi}{(\MUL-\xi)\EL} \cdot t
\; = \; \frac{\MUL-1}{\MUL} + \frac{1}{\MUL^2} \cdot \frac{t}{\EL}.
\end{align*}
Hence, the attacker's utility is upper-bounded by
$\VAL \cdot t \cdot (1-\Cdf{t}) 
\; \leq \;
\VAL \cdot 
t \cdot \left(\frac{1}{\MUL} - \frac{1}{\MUL^2} \cdot \frac{t}{\EL} \right)$.
This is maximized at $t^*=\frac{\MUL \EL}{2}$, so we obtain that
$\VAL \cdot t \cdot (1-\Cdf{t}) 
\; \leq \;
\VAL \cdot t^* \cdot \left(\frac{1}{\MUL} - \frac{1}{\MUL^2} \cdot \frac{t^*}{\EL} \right)
\; = \; \VAL \cdot \frac{\EL}{4}$.
\end{proof}

\begin{extraproof}{Theorem~\ref{thm:optimal-distributions}}
To complete the proof of Theorem~\ref{thm:optimal-distributions},
we now consider multiple targets $i$.
By the assumptions of the theorem and 
Proposition~\ref{lem:optimal-distributions}, 
against the proposed class of random sequences,
the attacker can obtain utility at most \quarter, 
regardless of which target $i$ he attacks and for how long;
this follows by substituting $\EL[i] = 1/\TFr[i] = 1/\Value{i}$.

We will show that no shift-invariant mixed defender schedule 
(now considered over the non-negative real numbers) can
achieve an expected attacker payoff strictly smaller than \quarter.
Focus on a shift-invariant mixed defender schedule \LOCDIST.
By Lemma~\ref{lem:canonical}, we may assume that \LOCDIST is
canonical.

Fix some index $i$ such that $\Value{i}/\TFr[i] \geq 1$.
Such an index must exist because $\sum_i \Value{i} = 1$ and 
$\sum_i \TFr[i] \leq 1$.
Because we assumed that $\Value{i} \leq \half$ for all $i$, 
this also implies that $\TFr[i] \leq \half$.

By Proposition~\ref{lem:attacker-minimum}, 
attacking target $i$ for $t = \EL[i]/2$ units of time, 
the attacker can guarantee himself a utility of at least  
\[
\Value{i} \cdot \frac{1-\Cdf[i]{0}}{4} \cdot \EL[i]
= \Value{i} \cdot \frac{(1-\Cdf[i]{0})^2}{4(\TFr[i]-\Cdf[i]{0})}
\stackrel{\TFr[i] \leq \half}{\geq} \Value{i} \cdot \frac{1}{4\TFr[i]}
\geq \quarter,
\]
where the final inequality followed because the chosen index $i$
satisfied $\Value{i}/\TFr[i] \geq 1$.
Hence, the attacker can guarantee himself a payoff of at least
\quarter against any mixed defender schedule, proving optimality of
the proposed class of random sequences.

Finally, we show that this applies to $2$-quasi-regular random sequences.
Assume that there exists a $b$ such that
$\Prob{b \leq \TBV[i] \leq 2b}=1$, and define $\MUL[i] = 2b/\EL[i]$.
First, this definition directly implies that
$\TBV[i] \leq \MUL[i] \EL[i]$ with probability 1.
Second, because $\TBV[i] \leq 2b$ with probability 1,
we get that $\EL[i] \leq 2b$, and hence
$\frac{\MUL[i]}{\MUL[i]+1}\EL[i] = \frac{2b}{2b/\EL[i]+1} \leq b$.
Hence, the fact that $\TBV[i] \geq b$ with probability 1
implies that $\TBV[i] \geq \frac{\MUL[i]}{\MUL[i]+1}\EL[i]$
with probability 1,
completing the proof.
\end{extraproof}

\section{An Optimal Defender Strategy} 
\label{sec:optimal-schedule}

In this section, we present Algorithm~\ref{alg:optimal},
constructing a $2$-quasi-regular random sequence.
Such a random sequence is optimal for the defender by
Theorem~\ref{thm:optimal-distributions}.
We begin with the high-level algorithm,
and fill in the details of the key steps below.
  
\begin{algorithm}[htb]
\caption{An optimal schedule for the defender\label{alg:optimal}}
\begin{algorithmic}[1]
\STATE Let $p_i = \Value{i}$ for all $i$.  
\STATE For each $i$, let $m_i$ be such that $2^{-m_i} \leq p_i <
2^{1-m_i}$.  Let $I_i = [2^{-m_i}, 2^{1-m_i}]$.
\STATE Use the algorithm from the proof of
Lemma~\ref{lem:transform-distribution} for $\bm{p}$ and the 
  $I_i$ to randomly round $\bm{p}$ to a probability vector $\bm{q}$,
  such that all but at most one index $i$ have $q_i = 2^{-m_i}$ or
  $q_i = 2^{1-m_i}$.  
\STATE Use the algorithm from the proof of
Lemma~\ref{lem:one-non-power} to produce a periodic sequence \SEQ.
\STATE Return the random sequence obtained by choosing a uniform
  random shift of \SEQ.
\end{algorithmic}
\end{algorithm}

Notice that the sequence produced by Algorithm~\ref{alg:optimal} is
shift-invariant by construction, but not ergodic, since it randomizes
over different shift-invariant distributions.

\begin{theorem} \label{thm:optimal-powers-of-2} 
The random sequence generated by Algorithm~\ref{alg:optimal} is
$2$-quasi-regular, and hence optimal for the defender.
\end{theorem}

We begin with a simple technical lemma.
\begin{lemma} \label{lem:subset-sum}
Let $S$ be a multiset of powers of 2, such that
$\max_{p \in S} p \leq 2^{-k} \leq \sum_{p \in S} p$.
Then, there exists a submultiset $T \subseteq S$
with $\sum_{p \in T} p = 2^{-k}$.
\end{lemma}

\begin{emptyproof}
We prove this claim by induction on $|S|$. 
The claim is trivial for $|S| = 1$.
Consider $|S| \geq 2$, and distinguish two cases.
\begin{enumerate}
\item If $S$ contains two copies of some number $p < 2^{-k}$,
then construct $S'$ by replacing these two copies with $p' = 2p$.
By induction hypothesis, $S'$ contains a subset $T'$ adding up to
$2^{-k}$. If $T'$ contained the newly constructed element $p'$, then
replace it with the two copies of $p$. 
In either case, we have the desired set $T \subseteq S$.
\item Otherwise, $S$ contains at most one copy of each number
$p \leq 2^{-k}$. 
If $S$ did not contain $2^{-k}$, then 
$\sum_{p \in S} p < \sum_{i=1}^\infty 2^{-(k+i)} = 2^{-k}$,
contradicting the assumptions of the lemma.
Hence, $S$ contains $2^{-k}$, and the singleton set of
that number is the desired subset. \QED
\end{enumerate}
\end{emptyproof}


\begin{lemma} \label{lem:transform-distribution}
Let $\bm{p} = (p_1, p_2, \ldots, p_n)$ be a probability distribution.
For each $i$, let $I_i = [\ell_i, r_i] \ni p_i$ be an interval.
Then, there exists a distribution $\mathcal{D}$
over probability distributions $\bm{q} = (q_1, q_2, \ldots, q_n)$ 
such that:
\begin{enumerate}
\item $\Expect{q_i} = p_i$ for all $i$,
\item $q_i \in I_i$ for all $\bm{q}$ in the support of $\mathcal{D}$, and
\item For each $\bm{q}$ in the support of $\mathcal{D}$, 
  all but at most one of the $q_i$ are equal to $\ell_i$ or $r_i$.
\end{enumerate}
\end{lemma}

\begin{proof}
We will give a randomized ``rounding'' procedure that starts with
$\bm{p}$ and produces a $\bm{q}$, satisfying all of the claimed
properties, by making the $p_i$ equal to $\ell_i$ or $r_i$ one at a time. 
The randomized rounding bears similarity to 
dependent randomized rounding algorithms in the approximation
algorithms literature (e.g.,
\cite{chekuri2010rounding,gandhi2006rounding,srinivasan2001level}),
though we do not require concentration bounds, and allow one of the
$q_i$ to be an interior point of its interval.
In the rounding, we always consider two indices $i,j$ with
$p_i = \ell_i + \epsilon_i, p_j = \ell_j + \epsilon_j$, 
such that $0 < \epsilon_i < r_i - \ell_i, 0 < \epsilon_j < r_j - \ell_j$.
(That is, neither $p_i$ nor $p_j$ is on the boundary of its interval.)
We probabilistically replace them with $p'_i, p'_j$, such
that all of the following hold:
\begin{itemize}
\item At least one of $p'_i, p'_j$ is at the boundary of its interval.
\item $\ell_i \leq p'_i \leq r_i$ and
$\ell_j \leq p'_j \leq r_j$.
\item $p'_i + p'_j = p_i + p_j$.
\item $\Expect{p'_i} = p_i$ and $\Expect{p'_j} = p_j$.
\end{itemize}

The rounding terminates when there is at most one $p_i$ that is not at
the boundary of its interval; let $\bm{q}$ be the vector of
probabilities at that point.
By iterating expectations, we obtain that $\Expect{q_i} = p_i$ for all
$i$. 
The upper and lower bounds on $q_i$ are maintained inductively, 
and the termination condition ensures the third claimed property of $\bm{q}$.

So consider arbitrary $p_i, p_j$ as above.
Let $\delta_i = \min(\epsilon_i, r_j - \ell_j - \epsilon_j)$ and
$\delta_j = \min(\epsilon_j, r_i - \ell_i - \epsilon_i)$.
With probability $\frac{\delta_j}{\delta_i+\delta_j}$, round
$p_i$ to $p'_i = p_i - \delta_i$ and $p_j$ to $p'_j = p_j + \delta_i$.
With probability 
$1 - \frac{\delta_j}{\delta_i+\delta_j} = \frac{\delta_i}{\delta_i+\delta_j}$,
round $p_i$ to $p'_i = p_i + \delta_j$ and $p_j$ to $p'_j = p_j - \delta_j$.

First, it is clear that $p'_i + p'_j = p_i + p_j$. 
Also, by definition of $\delta_i, \delta_j$, 
we get that $\ell_i \leq p'_i \leq r_i$
and $\ell_j \leq p'_j \leq r_j$.
If we round according to the first case, then
$p'_i = p_i - \delta_i$ and $p'_j = p_j + \delta_i$.
If $\delta_i = \epsilon_i$, then we get that $p'_i = \ell_i$,
while if $\delta_i = r_j - \epsilon_j$, then
$p'_j = \ell_j + \epsilon_j + (r_j - \ell_j - \epsilon_j) = r_j$.
Calculations are similar in the other case.
Finally, 
$\Expect{p'_i} 
\; = \; \frac{\delta_j}{\delta_i+\delta_j} \cdot (p_i - \delta_i) 
      + \frac{\delta_i}{\delta_i+\delta_j} \cdot (p_i + \delta_j) 
\; = \; p_i$.
Hence, all the claimed properties hold in each step.
\end{proof}

As a first step towards a 2-quasi-regular random sequence, 
we consider the case of probability vectors in which all probabilities
are powers of 2.\footnote{Lemma~\ref{lem:powers-of-two-schedule} generalizes to
  powers of any integer, and in fact to any probabilities $p_i$ such
  that for any $i,j$, we have that $p_i | p_j$ or $p_j | p_i$. 
  The existence of a schedule for powers of 2 (and the generalization)
  has been previously observed in the context of the Pinwheel Problem
  in \cite{HMRTV:pinwheel-conference}.}

\begin{lemma} \label{lem:powers-of-two-schedule}
 Assume that the probability vector $\bm{p}$ is such that each 
$p_i = 2^{-m_i}$ is a power of 2.
Then, there exists a regular sequence for $\bm{p}$.
\end{lemma}

\begin{proof}
  We will prove this claim by induction on the number of targets.  If
  we have a single target, then its probability must be 1, so it is
  visited at intervals of 1 and we set $\SEQ$ to be the constant
  sequence.  Otherwise, the maximum probability of any target is
  \half, and the sum of all probabilities is 1.
  Lemma~\ref{lem:subset-sum} therefore guarantees the existence of a
  subset $S$ whose probabilities add up to \half.

  Consider instances obtained from $S$ and $\bar{S}$ by scaling up all
  probabilities by a factor of 2, resulting in $p'_i = 2p_i$.  By
  induction hypothesis, each of those instances can be scheduled such
  that each target $i$ is visited every $1/p'_i = 1/(2p_i)$ time
  steps.  Now alternate between the two sequences.  In this new
  sequence, each target $i$ is visited every $2/p'_i = 1/p_i$ steps,
  as desired.
\end{proof}

Next, we show that sufficiently good sequences can also be achieved
when at most one of the probabilities is not a power of 2.

\begin{lemma} \label{lem:one-non-power} Assume that the probability
  vector $\bm{p}$ is such that each $p_i = 2^{-m_i}$ is a power of 2,
  except for (possibly) $p_1 = 2^{-m_1} - \epsilon$, with
  $0 \leq \epsilon < 2^{-(m_1+1)}$.  Then, there exists a (non-random)
  periodic sequence \SEQ with the following properties:
  \begin{enumerate}
  \item The time between consecutive visits to target $i > 1$ is
    always exactly $1/p_i$.
  \item The time between consecutive visits to target 1 is always
    either $2^{m_1}$ or $2^{m_1+1}$.
\item The frequency of target $i$ is $p_i$ for all $i$.
\end{enumerate}
\end{lemma}

\begin{proof}
We distinguish two cases\footnote{We would like to thank an anonymous
reviewer for suggesting this elegant proof.}:
\begin{enumerate}
\item If $p_1 \leq \half$, then by Lemma~\ref{lem:subset-sum},
there exists a subset $S \subseteq \SET{2, \ldots, n}$
with $\sum_{i \in S} p_i = \half$.
Schedule the targets in $S$ regularly in the odd time slots, 
and set $p'_i = 2p_i$ for all $i \notin S$.
The $p'_i$ satisfy the conditions of the lemma, so we inductively find
a schedule for all the $i \notin S$ satisfying the conclusion of the
lemma, then stretch this schedule by a factor of 2 and
use all the even time slots for it.
\item If $p_1 > \half$, then we schedule target 1 in all odd time
  slots, and set $p'_1 = 2(p_1 - \half)$ 
  and $p'_i = 2p_i$ for all $i \geq 2$.
  The new $p'_i$ sum to 1 and satisfy the conditions of the lemma.
  So we inductively find a schedule for the frequencies $p'_i$,
  stretch it by a factor of 2, and use all the even time slots for it.
  Notice that item 1 is scheduled at distance at most 2, and at least
  1, while all other items are scheduled regularly.
\end{enumerate}
That the resulting schedule is periodic is seen inductively
  over the applications of the two cases.
\end{proof}

\begin{extraproof}{Theorem~\ref{thm:optimal-powers-of-2}}
  Consider any target $i$.  The rounding of
  Lemma~\ref{lem:transform-distribution} guarantees that
  $2^{-m_i} \leq q_i \leq 2^{1-m_i}$.  Therefore, the algorithm of
  Lemma~\ref{lem:one-non-power} produces a random sequence
  $\SEQDIST_{\bm{q}}$ in which the time intervals between consecutive
  occurrences of target $i$ lie between $2^{m_i-1}$ and
  $2^{m_i}$. Thus to verify that $\SEQDIST_{\bm{q}}$ is
  $2$-quasi-regular is remains to show that the expected
  density of each target is equal to \Value{i}.
  But this is guaranteed by (1) in Lemma~\ref{lem:transform-distribution}.
  The optimality of $\SEQDIST_{\bm{q}}$ now follows from 
  Theorem~\ref{thm:optimal-distributions}.
\end{extraproof}

The second part of Theorem~\ref{thm:optimal-powers-of-2} shows
that 2-quasi regular random sequences exist; here, we remark that this
result cannot be improved, in the following sense 
(proved in Appendix~\ref{app:tightness}).

\begin{proposition} \label{thm:2-quasi-tight}
  Let $n=3$ and $\VAL=(1/2,1/3,1/6)$. Then, for every $\epsilon>0$,
  there are no $(2-\epsilon)$-quasi-regular random sequences.
\end{proposition}

\section{Golden Ratio Scheduling} 
\label{sec:golden-ratio}

In this section, we present a very simple ergodic random sequence.
The associated schedule may in general be suboptimal, but we prove that
it is within less than 0.6\% of optimal.

Let $\GR = \half(1+\sqrt{5})$ denote the Golden Ratio, 
solving $\GR^2 = \GR + 1$.
Given a desired frequency vector \FreqVec
(which will equal the targets' values, $\Freq{i} = \Value{i}$),
we identify the unit circle with $[0,1)$,
and equip it with addition modulo 1.
We define the function $\CIRC \colon [0,1) \to \SET{1, \ldots, n}$
via $\circinv{i} = [\sum_{i'<i} \Freq{i'}, \Freq{i} + \sum_{i' < i} \Freq{i'})$;
that is, we assign consecutive intervals of length \Freq{i} for the
targets $i$.

\begin{algorithm}[htb]
\caption{The Golden Ratio Schedule\label{alg:golden-ratio}}
\begin{algorithmic}[1]
\STATE Let \OFFSET be uniformly random in $[0,1)$.
\FOR{$t=0,1,2,\ldots$}
\STATE In step $t$, set $\Seq{t} = \circmap{(\OFFSET + \GR t) \mod 1}$.
\ENDFOR
\end{algorithmic}
\end{algorithm}

We can think of advancing a ``dial'' by \GR (or $\GR-1$)
at each step,
and visiting the target whose interval the dial falls into. 
This algorithm is nearly identical to one previously proposed for
hashing~\cite[pp.~510,511,543]{knuth1998art} and broadcast channel
sharing \cite{itai1984golden,panwar1992golden}.  
While the algorithm is simple, as stated, it seems to require precise
arithmetic with real numbers.
This issue is discussed in more detail in Appendix~\ref{sec:finite-precision}.

That Algorithm~\ref{alg:golden-ratio} returns an ergodic random
sequence follows from the classical fact that the action on the
interval by an irrational rotation is ergodic.
Our main theorem in this section is the following:

\begin{theorem} \label{thm:golden-ratio}
  The Golden Ratio algorithm is a 
  $\frac{2966-1290\sqrt{5}}{81} \approx 1.00583$
  approximation for the defender.
\end{theorem}
The underlying reason that this
schedule performs so well, and the reason for choosing
specifically the Golden Ratio,
is related to the hardness of diophantine approximation of the
Golden Ratio: it is an irrational number that is hardest to
approximate by rational numbers (see, e.g.,~\cite{hardy:wright}).

Our analysis relies heavily on various properties of Fibonacci numbers.  
We denote the \Kth{k} Fibonacci number by $\FIB{k}$,
indexed as
$\FIB{0} = 0, \FIB{1} = 1$ and $\FIB{k+2} = \FIB{k}+\FIB{k + 1}$.
The following basic facts about Fibonacci numbers are well-known, 
and easily proved directly or by induction.

\begin{lemma} \label{lem:fibonacci-basic}
\begin{enumerate}
\item For any $k$, we have that
$\FIB{k+2} \FIB{k} - \FIB{k+1}^2 = (-1)^{k+1}$.
\item \label{sqrt-five}
For any $k$, we have that
$\FIB{k} = \frac{\GR^k - (-1/\GR)^k}{\sqrt{5}}$.
\item For any odd $k$, we have that $\FIB{k+1}/\FIB{k} < \GR$.
\item For any even $k$, we have that $\FIB{k+1}/\FIB{k} > \GR$.
\end{enumerate}
\end{lemma}

To prove Theorem~\ref{thm:golden-ratio}, we analyze the distribution
of \TBV[i] for any target $i$.
The proof of Theorem~\ref{thm:golden-ratio} consists of two parts.
First, Theorem~\ref{THM:THREE-FIBONACCI} precisely characterizes the
distribution of \TBV[i] for every target $i$, 
i.e., it characterizes exactly, for each target $i$ and
$\tau$, how frequently a visit to target $i$ is followed by another
visit $\tau$ steps later.
(We call such a $\tau$ a \emph{return time}.)
As a second part, we characterize the attacker's best response against
this distribution, and calculate its cost to the defender.

\begin{theorem}[Slater \cite{slater1950gaps}]\label{THM:THREE-FIBONACCI}
Assume that $\FFR \leq \half$.
Let $k$ be smallest such that 
\begin{align}
|\FIB{k+1} - \GR \FIB{k}| \; \leq \; \GR \FFR
\label{eqn:ell-definition}
\end{align}

Then, the distribution of return times is
\begin{equation} \label{eqn:return-probs}
  \begin{split}
    \Prob{\TBV[i] = \FIB{k+1}} & = \FIB{k+1} \cdot \left(\FFR -
      (1/\GR)^{k+1} \right),\\
    \Prob{\TBV[i] = \FIB{k+2}} & = \FIB{k+2} \cdot \left(\FFR -
      (1/\GR)^{k+2} \right),\\
    \Prob{\TBV[i] = \FIB{k+3}} & = \FIB{k+3} \cdot \left(-\FFR +
      (1/\GR)^{k} \right),\\
    \Prob{\TBV[i] = t} & = 0 \mbox{ for all other } t.
\end{split}
\end{equation}
\end{theorem}

Theorem~\ref{THM:THREE-FIBONACCI} shows, remarkably, that for each
possible \FFR, there are at most three possible return times, 
and they are three consecutive Fibonacci numbers. 
Theorem~\ref{THM:THREE-FIBONACCI} is a special case of a theorem of
Slater \cite[Theorem 4]{slater1950gaps} (see also \cite{slater1967gaps}),
which characterizes the distribution when the Golden Ratio \GR is
replaced by an arbitrary real number.
We give a self-contained proof for the simpler case of the Golden
Ratio in Appendix~\ref{sec:slater-proof}.

As a direct corollary of Theorem~\ref{THM:THREE-FIBONACCI}, 
we obtain an upper bound on the quasi-regularity of the 
Golden Ratio schedule.

\begin{corollary} \label{cor:golden-ratio-regular}
The Golden Ratio schedule is $3$-quasi-regular.
If $\Freq{i} \leq 1 - 1/\GR$ for all $i$, then it is $8/3$-quasi-regular.
As the frequencies $\Freq{i} \to 0$, 
the regularity guarantee improves to $\GR^2$-quasi-regular.
\end{corollary}

\begin{proof}
Consider one target $i$ with desired frequency \Freq{i},
and define $k$ as in Theorem~\ref{THM:THREE-FIBONACCI}.
The schedule is at worst $(\FIB{k+3}/\FIB{k+1})$-quasi-regular.
For all $k$, we have the bound $\FIB{k+3}/\FIB{k+1} \leq 3$.
If $\Freq{i} < 1-1/\GR$, then $k \geq 2$,
and for $k \geq 2$,
the ratio $\FIB{k+3}/\FIB{k+1}$ is upper-bounded by $8/3$,
converging to $\GR^2$ as $k \to \infty$.
\end{proof}

\subsection{The Optimal Attacker Response to 3-Point Distributions}
Next, we characterize the optimal attacker response to defender
strategies in which the return time distribution to target $i$ is
supported on three points $1 \leq x_1 < x_2 < x_3$ only, so that with
probability one, $\TBV[i] \in \{x_1,x_2,x_3\}$. In the remainder of
this section, we omit the subscript $i$, as we will always analyze one
target only.

Recall that $\EL = 1/\Cdf{1}$, and let
$q_j = \Prob{\TBV=x_j}\cdot\frac{\EL}{x_j}$. Note that
$\sum_j q_j x_j = \sum_j \Prob{\TBV=x_j} \cdot \EL = \EL$.  Informally
(but in a sense that can be formalized), the $q_j$'s are the return
time probabilities from the point of view of the defender, rather than
those of the attacker at a fixed time $0$.  (The attacker's
distribution of \TBV[i] oversamples long return times, compared to the
defender's distribution.)  The attacker's response and utility are
summarized by the following lemma.

\begin{lemma}
  \label{LEM:ATTACKER-THREE-POINTS}

  Assume that $x_2 \leq 2x_1, x_3 \leq 2x_2$ (as is the case with
  Fibonacci numbers).  Let $u^*_1 = \quarter \cdot \VAL \cdot \EL$,
  and
  $u^*_2 = \quarter \cdot \VAL \frac{(\EL - q_1 x_1)^2}{\EL(1-q_1)}$.
  Then, against the given three-point distribution, the attacker's
  utility is at most $\max(u^*_1, u^*_2)$.  Against \emph{any}
  distribution with expected defender absence \EL, the attacker's
  utility is at least $u^*_1$.
\end{lemma}
\begin{proof}
The lower bound of $u^*_1$ is simply the statement of 
Corollary~\ref{lem:sufficient-optimal}.
So we focus on the upper bound of $\max(u^*_1, u^*_2)$ for the
remainder of the proof.

From the attacker's perspective, when arriving at a target, 
by Equations~\eqref{eqn:return-probs} and~\eqref{eq:cdf-tbv},
the CDF of the distribution of the defender's next return time is
\begin{align*}
\Cdf{t} & = 
\begin{cases}
\frac{1}{\EL} \cdot t & \mbox{ for } t \leq x_1,\\
\frac{1}{\EL} \cdot (t(1-q_1) + q_1 x_1) & \mbox{ for } x_1 \leq t \leq x_2,\\
\frac{1}{\EL} \cdot (t(1-q_1-q_2) + q_1 x_1 + q_2 x_2) & \mbox{ for } x_2 \leq t \leq x_3,\\
1 & \mbox{ for } t \geq x_3.
\end{cases}
\end{align*}

Since the attacker's utility for waiting for $t$ steps is $t(1-F(t))$,
$t \geq x_3$ cannot be optimal for him.
By taking derivatives with respect to $t$, we obtain the following
local optima for the functions in the remaining three cases:

\begin{align*}
t^*_1 & = \frac{\EL}{2}
& t^*_2 & = \frac{\EL - q_1 x_1}{2(1-q_1)} 
& t^*_3 & = \frac{\EL - q_1 x_1 - q_2 x_2}{2(1-q_1-q_2)} = \frac{x_3}{2}. 
\end{align*}

These are all local maxima because the functions are concave. 
Under the assumption of the theorem, $t^*_3 = x_3/2 \leq x_2$.
Therefore, $t^*_3$ does not lie in the interval it optimized for,
and can never be optimal.
As a result, the attacker's best response\footnote{It is possible that
one of $t^*_1, t^*_2$ also lies outside its interval. 
But our goal here is only to derive an upper bound.}
will always be $t^*_1$ or $t^*_2$.
The attacker's utility for these two attack times will be

\begin{align*}
u^*_1 & = \VAL \cdot t^*_1 \cdot (1-F(t^*_1)) 
      \; = \; \frac{\VAL \EL}{2} \cdot \half 
      \; = \; \frac{\VAL \EL}{4},\\
u^*_2 & = \VAL \cdot t^*_2 \cdot (1-F(t^*_2)) 
      \; = \; \VAL \cdot \frac{\EL - q_1 x_1}{2(1-q_1)} \cdot 
              \left(1-\frac{1}{\EL} \cdot \left(\frac{\EL - q_1 x_1}{2} + q_1 x_1\right)\right) \\
&     = \;  \frac{\VAL}{4} \cdot \frac{(\EL - q_1 x_1)^2}{\EL(1-q_1)}
      \; = \;  \frac{\VAL}{4} \cdot \frac{(q_2 x_2 + q_3 x_3)^2}{\EL(q_2+q_3)}.
\end{align*}

The attacker's utility will thus be at most the maximum of $u^*_1,
u^*_2$.
\end{proof}

\subsection{The Attacker's Response to the Golden Ratio Schedule}

\begin{emptyextraproof}{Theorem~\ref{thm:golden-ratio}}
Applying Lemma~\ref{LEM:ATTACKER-THREE-POINTS},
it is our goal to upper-bound $\frac{\max(u^*_1,u^*_2)}{u^*_1}$.
If $\max(u^*_2,u^*_1) = u^*_1$, the approximation ratio is 1;
hence, it suffices to upper-bound $u^*_2/u^*_1$.

In applying Lemma~\ref{LEM:ATTACKER-THREE-POINTS},
we have $\EL = 1/\VAL, x_1 = \FIB{k+1}, x_2 = \FIB{k+2}, x_3 = \FIB{k+3}$,
and the $q_j$ are given via Equation~\eqref{eqn:return-probs}
and $q_j = \Prob{\TBV[i]=x_j}\cdot\frac{\EL}{x_j}$.
Then, $u^*_2/u^*_1$ can be written as follows:

\begin{align}
\nonumber \frac{(\EL - q_1 x_1)^2}{\EL^2(1-q_1)}
& = \VAL^2 \cdot \frac{(q_2 \FIB{k+2} + q_3 \FIB{k+3})^2}{q_2+q_3}\\ \nonumber
& = \VAL \cdot \frac{\left(  
       \left(\VAL - (1/\GR)^{k+2} \right) \cdot \FIB{k+2} 
     + \left(-\VAL + (1/\GR)^{k} \right) \cdot \FIB{k+3} \right)^2}{
       \left(\VAL - (1/\GR)^{k+2} \right)
     + \left(-\VAL + (1/\GR)^{k} \right)}\\ \nonumber
& = \VAL \cdot \frac{\left(  
   (1/\GR)^{k+2} \cdot (\GR^2-1) \cdot \FIB{k+2}
   + \left(-\VAL + (1/\GR)^{k} \right) \cdot \FIB{k+1}
\right)^2}{(1/\GR)^{k+2} \cdot (\GR^2-1)}. \\ 
& = \VAL \cdot (1/\GR)^{k+1} \cdot 
\left(   \FIB{k+2} + \GR \FIB{k+1} 
       - \VAL \cdot \GR^{k+1} \cdot \FIB{k+1} \right)^2.
\label{eqn:fibonacci-ratio}
\end{align}

Treating everything except \VAL as a constant, the approximation ratio
is thus of the form $g(\VAL) = a \VAL \cdot (c - b\VAL)^2$.  $g$ has a
local maximum of $4ac^3/27b$ at $\VAL = c/(3b)$, a local minimum of
$0$ at $\VAL = c/b$, and goes to infinity as $\VAL \to \infty$.  Thus,
the two candidates for \VAL that we need to check are (1) the largest
\VAL that is possible for a given $k$, and (2) the value
$\VAL = c/(3b)$.

We therefore next calculate the
largest possible \VAL for a given $k$.
By recalling the definition of $k$ from
Equation~\eqref{eqn:ell-definition} 
(smallest such that $|\FIB{k+1}/\GR - \FIB{k}| \leq \VAL$),
and using Lemma~\ref{lem:fibonacci-difference},
we can solve for \VAL to determine the range in which we obtain a
particular $k$, giving us that 
$\VAL \in \left[(1/\GR)^{k+1}, (1/\GR)^{k} \right]$.

\begin{enumerate}
\item If we substitute the upper bound
$\VAL = (1/\GR)^{k}$, 
Equation~\eqref{eqn:fibonacci-ratio} simplifies to
\begin{align*}
& \VAL \cdot (1/\GR)^{k+1} \cdot 
\left(  \FIB{k+2} + \GR \FIB{k+1} 
       - (1/\GR)^{k} \cdot \GR^{k+1} \cdot \FIB{k+1} \right)^2\\
& = (1/\GR)^k \cdot (1/\GR)^{k+1} \cdot 
\left(  \FIB{k+2} + \GR \FIB{k+1} - \GR \FIB{k+1} \right)^2\\
& = \frac{1}{5} \cdot (1/\GR)^{2k+1} \cdot 
\left( \GR^{2k+4} - 2\GR^{k+2} (-1/\GR)^{k+2} + (-1/\GR)^{2k+4} \right)\\
& = 
\frac{1}{5} \cdot (1/\GR)^{2k+1} \cdot 
\left( \GR^{2k+4} - 2 (-1)^k + (1/\GR)^{2k+4} \right)\\
& \leq \frac{1}{5} \cdot (\GR^3 + 3/\GR^{2k+1}) 
\; \leq \; \frac{1}{5} \cdot (\GR^3 + 3/\GR^3) 
\; < \; 1.
\end{align*}
This shows that the attacker's utility cannot be maximized by waiting
for more than $x_1$ steps when \VAL is as large as it can be for a
given $k$.

\item Next, we investigate the local maximum\footnote{%
    This local maximum is indeed always a feasible choice for \VAL for
    a given $k$, but since we are only interested in an upper bound,
    we omit the feasibility proof.}  of
  Equation~\eqref{eqn:fibonacci-ratio}. 
Substituting $a = (1/\GR)^{k+1}$,
$b = \GR^{k+1} \cdot \FIB{k+1}$, and
$c = \FIB{k+2} + \GR \FIB{k+1}$, the approximation ratio is

\begin{align*}
& \frac{4}{27} \cdot (1/\GR)^{2k+2} \cdot 
\frac{\left(\FIB{k+2} + \GR \FIB{k+1} \right)^3}{\FIB{k+1}}\\
& = \frac{4}{27} \cdot \frac{1}{5} \cdot 
(1/\GR)^{2k+2} \cdot 
\frac{\left(\GR^{k+2} - (-1/\GR)^{k+2} + \GR \cdot \GR^{k+1} - \GR
    \cdot (-1/\GR)^{k+1} \right)^3}{\GR^{k+1} - (-1/\GR)^{k+1}}\\
& = \frac{4}{27} \cdot \frac{1}{5} \cdot 
(1/\GR)^{2k+2} \cdot 
\frac{\left(2\GR^{k+2} - (-1/\GR)^{k+1}\right)^3}{\GR^{k+1} - (-1/\GR)^{k+1}}.
\end{align*}

We will approximate the function 
$\frac{\left(2\GR^{k+2} - (-1/\GR)^{k+1}\right)^3}{\GR^{k+1} - (-1/\GR)^{k+1}}$
by $8 \, \GR^{2k+5}$, its highest-order term. We therefore consider 
$\frac{\left(2\GR^{k+2} - (-1/\GR)^{k+1}\right)^3}{\GR^{k+1} - (-1/\GR)^{k+1}}/
(8 \GR^{2k+5})$. 
When $k$ is even, this ratio is always upper-bounded by 1 (and
increasing in $k$, converging to 1).
When $k$ is odd, this ratio is \emph{lower-bounded by 1}, and
decreasing in $k$, also converging to 1.
Thus, it is maximized among feasible values of $k$ for $k=3$,
where it equals $\frac{8\GR^{15}-12\GR^6+6\GR^{-3}-\GR^{-12}}{8\GR^{15}-8\GR^7}$.
Overall, we get an upper bound on the attacker's utility of
\begin{align*}
& \frac{1}{5} \cdot \frac{4}{27} \cdot 
\frac{8\GR^{15}-12\GR^6+6\GR^{-3}-\GR^{-12}}{8\GR^{15}-8\GR^7} \cdot
(1/\GR)^{2k+2} \cdot (8\GR^{2k+5})\\
& = \frac{\GR^3}{5} \cdot \frac{4}{27} \cdot
\frac{8\GR^{15}-12\GR^6+6\GR^{-3}-\GR^{-12}}{\GR^{15}-\GR^7}.
\end{align*}
To evaluate this ratio, we can repeatedly apply the fact that
$\GR^2 = 1+\GR$, then substitute that $\GR = \frac{1+\sqrt{5}}{2}$,
make the denominator rational, and cancel out common factors.
This shows that
$\frac{\GR^3}{5} \cdot \frac{4}{27} \cdot
\frac{8\GR^{15}-12\GR^6+6\GR^{-3}-\GR^{-12}}{\GR^{15}-\GR^7}
= \frac{2966-1290\sqrt{5}}{81} \approx 1.00583$, 
completing the proof.\QED
\end{enumerate}
\end{emptyextraproof}

\begin{remark}
The analysis of Theorem~\ref{thm:golden-ratio} is actually tight.
By choosing $k=3$,
$b = \GR^{4} \cdot \FIB{4}$, and
$c = \FIB{5} + \GR \FIB{4}$, 
we obtain that the worst-case value of a target is
$\VAL = c/(3b) = \frac{23}{18} - \frac{\sqrt{5}}{2} \approx 0.1597$.
Substituting $k=3$ into the attacker's utility in Case (2) (before the
lower bound is applied) gives us exactly a ratio of
$\frac{2966-1290\sqrt{5}}{81} \approx 1.00583$.
\end{remark}

\section{Scheduling via Matching} 
\label{sec:matching}

In this section, we show that for every $\epsilon > 0$,
if the individual probabilities \Value{i} are small enough (as a
function of $\epsilon$), 
then $(1+\epsilon)$-quasi-regular ergodic, periodic schedules exist.

In order to obtain a periodic strategy, it is clearly necessary for
all target values \Value{i} (equaling the visit frequencies) to be
rational. 
Write $\Value{i} = a_i/b_i$, and let $M = \lcm(b_1, \ldots, b_n)$.
Our algorithm is based on embedding $M$ slots for visits evenly on the
unit circle, and matching them with targets to visit.
We identify the circle with the interval $[0,1]$ and use the 
distance $d(x,y) = \min(|x-y|, 1-|x-y|)$.

\begin{algorithm}[htb]
\caption{A matching-based algorithm for a periodic defender strategy\label{alg:matching}}
\begin{algorithmic}[1]
\FOR{each target $i$}
\STATE Let $\theta_i \in [0,1]$ independently uniformly at random.
\STATE Let $A_i = M \cdot \Value{i}$.
\STATE For $j = 0, \ldots, A_i - 1$, 
let $\Pos{i}{j} = (\theta_i + j/A_i) \mod 1$.
\ENDFOR
\STATE Let $\delta = \frac{1}{M} \cdot \sqrt{\frac{n \log M}{2}}$.
\STATE Let $\ALLSLOTS = \SET{0, 1, \ldots, M-1}$ be the set of 
\emph{slots} and let $\POS = \Set{(i,j)}{0 \leq j < A_i}$.
Define a bipartite graph $G$ on $\ALLSLOTS \cup \POS$
by including an edge between $\Slot \in \ALLSLOTS$ and
$(i,j) \in \POS$ iff
$d(\Pos{i}{j}, \Slot/M) \leq \delta$.
\IF{$G$ contains a perfect matching $\mathcal{M}$}
\STATE Define a sequence \SEQ with period $M$ as follows: For each
time \Slot, set \Seq{\Slot} to be the
(unique) target $i$ such that \Slot is matched with $(i,j)$ in
$\mathcal{M}$ for some $j$.
\ELSE
\STATE Start from the beginning.
\ENDIF
\end{algorithmic}
\end{algorithm}

\begin{theorem} \label{thm:matching}
Fix $\epsilon > 0$, and assume that
$\Value{i} \leq \frac{\epsilon}{4+2\epsilon} \cdot \sqrt{\frac{2}{n \log M}}$
for all $i$.
Then, Algorithm~\ref{alg:matching} succeeds with
high probability.
Whenever Algorithm~\ref{alg:matching} succeeds, it produces a
$(1+\epsilon)$-quasi-regular (and hence defender-optimal) sequence.
\end{theorem}

We begin by proving the second part of the theorem.
First, in a perfect matching, exactly $A_i$ of the $M$ slots, 
i.e., an \Value{i} fraction, are scheduled for target $i$, 
giving that $\TFr[i] = a_i/b_i = \Value{i}$.
Thus, $\EL[i] = 1/\Value{i}$.

If \Slot is matched with $(i,j)$, by definition of the edges, 
$d(\Pos{i}{j}, \Slot/M) \leq \delta$.
Consider two occurrences $j,j'$ of target $i$,
and let $\Slot, \Slot'$ be the slots they are matched to.
Then, by triangle inequality, 
\begin{align}
d(\Slot/M, \Slot'/M)
& \geq d(\Pos{i}{j}, \Pos{i}{j'}) - 2\delta
\; \geq \; \frac{1}{A_i} - 2\delta.
\label{eqn:lower-distance}
\end{align}

On the other hand, specifically for consecutive occurrences of target
$i$, i.e., the slots matched to \Pos{i}{j} and \Pos{i}{j+1}, we get

\begin{align}
d(\Slot/M, \Slot'/M)
& \leq d(\Pos{i}{j}, \Pos{i}{j'}) + 2 \delta
\; \leq \; \frac{1}{A_i} + 2\delta.
\label{eqn:upper-distance}
\end{align}

Using that $\delta = \frac{1}{M} \cdot \sqrt{\frac{n \log M}{2}}$,
as well as $\frac{1}{A_i} = \frac{1}{M \cdot \Value{i}}$ and 
$\Value{i} \leq \frac{\epsilon}{4+2\epsilon} \cdot \sqrt{\frac{2}{n \log M}}$,
we obtain that
\begin{align*}
\frac{\frac{1}{A_i} + 2\delta}{\frac{1}{A_i} - 2\delta}
& = \frac{\frac{1}{\Value{i}} + 2 \sqrt{\frac{n \log M}{2}}}{%
\frac{1}{\Value{i}} - 2 \sqrt{\frac{n \log M}{2}}} 
\; \leq \; \frac{\frac{4+2\epsilon}{\epsilon} \cdot \sqrt{\frac{n \log
      M}{2}} + 2 \sqrt{\frac{n \log M}{2}}}{%
      \frac{4+2\epsilon}{\epsilon} \cdot \sqrt{\frac{n \log
      M}{2}} - 2 \sqrt{\frac{n \log M}{2}}}
\; = \; \frac{\frac{4+2\epsilon}{\epsilon} + 2}{\frac{4+2\epsilon}{\epsilon}- 2} 
\; = \; 1+\epsilon,
\end{align*}
proving that the resulting sequence is $(1+\epsilon)$-quasi-regular.

\smallskip

To complete the proof, it remains to show that with high probability,
the graph $G$ contains a perfect matching. 
We will prove this using Hall's Theorem and a direct application of
the Hoeffding Bound:

\begin{lemma}[Hoeffding Bound] \label{lem:hoeffding}
Let $X_i$ be independent random variables such that
$a_i \leq X_i \leq b_i$ with probability 1.
Let $X = \sum_i X_i$. Then, for all $t > 0$,
\[
\Prob{X < \Expect{X} - t}, \Prob{X > \Expect{X} + t}
 \; < \; \e^{-\frac{2t^2}{\sum_i (b_i-a_i)^2}}.
\]
\end{lemma}

To establish the Hall condition of $G$, we begin with intervals
$\SLOTS \subseteq \ALLSLOTS$ of slots, 
and then use the bounds for intervals to derive the
condition for arbitrary sets of slots.
A similar style of proof was used by 
Tijdeman~\cite{tijdeman1973distribution} to construct a schedule with somewhat
different specific combinatorial properties.

For any set $\SLOTS \subseteq \ALLSLOTS$ of slots, let \Neigh{\SLOTS}
denote the neighborhood of \SLOTS in $G$.  Fix an interval
$\SLOTS = \{\ell, \ell+1, \ldots, r-1\} \subseteq [0, M)$
with $\ell, r$ integers. 
Let the random variable $X_{\SLOTS} = |\Neigh{\SLOTS}|$ denote the
number of neighbors in $G$ of slots in the interval \SLOTS.
For each target $i$, let $X_{\SLOTS,i}$ be the number of $j$ such that
$(i,j) \in \Neigh{\SLOTS}$. 
Then, $X_{\SLOTS} = \sum_i X_{\SLOTS,i}$, 
and the $X_{\SLOTS,i}$ are independent.

\begin{lemma} \label{lem:one-type}
Fix a target $i$, and assume that $|\SLOTS| \leq (1-2\delta) M$, 
and write $x_i = A_i \cdot (2\delta + (r-\ell)/M)$. 
Then, $\Expect{X_{\SLOTS,i}} = x_i$
and $X_{\SLOTS,i} \in \SET{\Floor{x_i},\Floor{x_i}+1}$.
\end{lemma}

\begin{proof}
  For each slot $\Slot \in \SLOTS$, let $J_{\Slot}$ be the interval
  $[(\Slot/M-\delta) \mod 1, (\Slot/M+\delta) \mod 1]$.  Then, $(i,j)$
  is adjacent to \Slot iff $\Pos{i}{j} \in J_{\Slot}$.  Define
  $J := \bigcup_{\Slot \in \SLOTS} J_{\Slot}$; then, $(i,j)$ is
  adjacent to a slot in \SLOTS iff $\Pos{i}{j} \in J$.  Because
  $\delta = \frac{1}{M} \cdot \sqrt{\frac{n \log M}{2}} \geq 1/(2M)$,
  $J$ is an interval, of the form
  $[(\ell/M-\delta) \mod 1, (r/M+\delta) \mod 1]$.

  The length of the interval is $|J|=2\delta + (r-\ell)/M$.  Because
  each \Pos{i}{j} is uniformly random in $[0,1]$,
  $\Expect{X_{\SLOTS,i}}=A_i \cdot |J|$.  Furthermore, because
  $d(\Pos{i}{j}, \Pos{i}{j+1}) = 1/A_i$, there can be no more than
  $1+\Floor{\frac{|J|}{1/A_i}}$ pairs $(i,j)$ with $\Pos{i}{j} \in J$,
  and no fewer than $\Floor{\frac{|J|}{1/A_i}}$.  Finally, note that
  $\frac{|J|}{1/A_i} = A_i(2\delta + (r-\ell)/M) = x_i$.
\end{proof}

We use Lemma~\ref{lem:one-type} to show that with 
high probability, $G$ has a perfect matching.

\begin{lemma} \label{lem:matching} Whenever
  $\Value{i} \leq \frac{\epsilon}{4+2\epsilon} \cdot \sqrt{\frac{2}{n
      \log M}}$
  for all $i$, with probability at least $1-1/M^2$, $G$ contains a
  perfect matching.
\end{lemma}

\begin{proof}
First, we show that when the Hall condition holds for all
\emph{intervals} \SLOTS of slots, it holds for all \emph{sets} \SLOTS.
We prove this by induction on the number of disjoint intervals
that \SLOTS comprises. The base case of \SLOTS being an interval
is true by definition.
For the induction step, suppose that $k \geq 2$ and
$\SLOTS = \bigcup_{j=1}^k \SLOTS[j]$, where the 
\SLOTS[j] are disjoint intervals.

If the neighborhoods of all the \SLOTS[j] are disjoint,
then $|\Neigh{\SLOTS}| = \sum_j |\Neigh{\SLOTS[j]}| 
\geq \sum_j |\SLOTS[j]| = |\SLOTS|$,
where the inequality was from the base case (intervals).
Otherwise, w.l.o.g., 
$\Neigh{\SLOTS[k]} \cap \Neigh{\SLOTS[k-1]} \neq \emptyset$.
Then, there exists an interval 
$I' \supset \SLOTS[k] \cup \SLOTS[k-1]$ with
$\Neigh{I'} = \Neigh{\SLOTS[k]} \cup \Neigh{\SLOTS[k-1]}$.
Let $\SLOTSP = \SLOTS \cup I'$.
We get that 
$|\Neigh{\SLOTS}| = |\Neigh{\SLOTSP}| \geq |\SLOTSP| \geq |\SLOTS|$, 
where the first inequality was by induction hypothesis 
(because \SLOTSP has at least one less interval).

Next, we establish that the Hall Condition holds with high probability
for all $M^2$ intervals.
First, focus on one interval $\SLOTS = [\ell, r)$, with $\ell, r \in \N$.
If $|\SLOTS| > (1-2\delta) M$, then \Neigh{\SLOTS} contains all
pairs $(i,j)$, so the Hall Condition is satisfied.
So focus on $|\SLOTS| \leq (1-2\delta) M$.
From Lemma~\ref{lem:one-type}, we get that
\[
\Expect{X_{\SLOTS}} 
\; = \; \sum_i A_i \cdot (2\delta + (r-\ell)/M)
\; = \; 2 \delta M + (r-\ell)
\; = \; \sqrt{2n \log M} + (r-\ell).
\]

Furthermore, $X_{\SLOTS}$ is the sum of independent random variables
$X_{\SLOTS,i}$ which each takes on one of two adjacent values.  From
the Hoeffding Bound (Lemma~\ref{lem:hoeffding}), we get that
\[
\Prob{X_{\SLOTS} < (r-\ell) + 2\sum_i A_i \delta - \tau}
\; < \; \e^{-2\tau^2/n}.
\]

Because $|\SLOTS| = r - \ell$, 
choosing $\tau = 2\sum_i A_i \delta = 2M \delta = \sqrt{2n \log M}$, 
we get that
\[
|\Neigh{\SLOTS}| \; = \; X_{\SLOTS}
\; \geq \; (r-\ell) + 2\sum_i A_i \delta - \tau
\; = \; r-\ell,
\]
with probability at least $1-\e^{-4n \log M/n} = 1-1/M^{4}$.

Taking a union bound over all $M^2$ candidate intervals \SLOTS, 
we obtain that the probability of having a perfect matching is at least
$1-1/M^{2}$.
Thus, with high probability, $G$ contains a perfect matching.
This completes the proof of Lemma~\ref{lem:matching} and thus also
Theorem~\ref{thm:matching}.
\end{proof}


\section{Future Work} \label{sec:conclusion}

Our work suggests a number of directions for future work.
Most immediately, it suggests trying to find optimal
\emph{ergodic} schedules for all value vectors 
(not only those covered by Theorem~\ref{thm:matching}).
A promising approach toward this goal is to use the randomized
rounding of Section~\ref{sec:optimal-schedule}, 
but re-round the probabilities every $T$ steps,
for some sufficiently large ``epoch size'' $T$.
The difficulty with this approach is ``stitching together'' the
schedules for different rounded frequencies at the boundary of epochs,
without violating the conditions of Theorem~\ref{thm:optimal-distributions}.

Throughout, we assumed that no target had value more than the sum of
all other targets' values, i.e., $\Value{i} \leq \half$ for all $i$.
When this assumption is violated,
the optimal schedule will wait at the highest-value target.
In the specific case of two targets of values
$\Value{1} > \half$ and $\Value{2} = 1- \Value{1}$,
it is fairly straightforward to calculate that the wait time at target 1 is
$2(\sqrt{\frac{\Value{1}}{\Value{2}}}-1)$.
We anticipate that a similar analysis will extend to more
than two targets.
The difficulty is that the waiting time at one target will result in
qualitatively different schedules, likely to complicate the analysis.

We assumed here that the game is zero-sum.
In general, the utilities of the attacker and defender may be different.
A general treatment is likely quite difficult.
One special case is motivated directly by the wildlife protection
application, and appears quite amenable to analysis.
Specifically, when a poacher kills animals (or chops down trees),
even if the poacher is captured, the damage is not reversed.
Thus, while the attacker's utility is as before, the defender's
utility from visiting target $i$ at time $\tau$ when the attacker
intends to stay for $t$ units of time is
$- \Value{i} \cdot \min(\tau, t)$.
Because the sum of utilities is thus not constant
(and typically negative), the defender's goal becomes more strongly
that of deterring (rather than just capturing) the attacker.
As a result, the ability to commit to a strategy first
(i.e., treating the game as a
\emph{Stackelberg Game}~\cite{von1934marktform, conitzer2006commit})
may now carry some advantage for the defender.
One can show that in the model in this paragraph,
whenever the attacker attacks target
$i$ for $t \leq \EL[i]/2$ units of time, the defender's utility is
$- \frac{3}{2} \AUtilTwo{\CDF[i]}{i}{t}$.
Since the optimal defender strategies of
Section~\ref{sec:optimal-schedule} and \ref{sec:matching} ensure such
a choice of $t$ by the attacker, the algorithms in those sections are
optimal in the non-zero sum model as well.

Among the other natural generalizations are the attacker's (and
defender's) utility function and more complex constraints on the
defender's schedule.
Throughout, we have assumed that the attacker's utility grows linearly
in the time spent at a target.
The security game formulations studied in much of the prior work in
the area \cite{tambe2011security} correspond to a step function at 0:
when the attacker reaches an unprotected target, he immediately causes
the maximum target-specific damage \Value{i}
(e.g., by blowing up the target).
Other natural utility functions suggest themselves:
if the resources to collect at targets are limited,
the utility function would be linear with a cap.
If a destructive attack takes a non-zero amount of time to set up, 
one obtains a step function at a time other than 0.
The latter leads to a scheduling problem with a harder constraint on
the inter-visit absence time from targets $i$ --- as in some of the
prior security games literature, the defender may ``sacrifice'' some
low-value targets to be able to fully protect the others.

The other natural generalization is to relax the assumption of uniform
travel time between targets. 
If an arbitrary metric is defined between targets, the problem becomes
significantly more complex: even if all targets have value 1,
the attacker's utility will be proportional to the cost of a minimum TSP
tour, and thus the defender's optimization problem is NP-hard. 
However, it is far from obvious how to adapt standard TSP
approximation techniques to the general problem with non-uniform
values: high-value targets should be visited more frequently, and TSP
approximation algorithms are not suited to enforce constraints that
these visits be spaced out over time.

As with TSP problems and past work on security games, a further
natural generalization is to consider multiple defenders,
as, e.g., in~\cite{KorzhykCP11}.

\subsubsection*{Acknowledgments}
We would like to thank Omer Angel, Sami Assaf, Shaddin Dughmi, Fei
Fang, Ron Graham,
Bobby Kleinberg, Jim Propp, and Milind Tambe for useful discussions,
and anonymous referees for useful feedback.

\newpage
\appendix
\newcommand{\geom}{\text{geom}}
\newcommand{\unif}{\text{unif}}

\section{Utility of an i.i.d.\ Defender} 
\label{sec:geometric}

One of the most natural random sequences to consider is the 
\emph{i.i.d.}~one, in which at each step $t$, the defender visits
target $i$ with probability $p_i$, independent of any past choices.
Intuitively, this strategy is suboptimal because it may visit a target
$i$ several times in close succession, or go for a long time without
visiting target $i$.  
Here, we calculate the approximation ratio of this strategy, showing:

\begin{proposition} \label{lem:markovian}
The i.i.d.~strategy is a $4/\e$-approximation for the defender, and
this is tight.
\end{proposition}

\begin{proof}
From the attacker's viewpoint, 
the defender's next arrival time at target $i$ is the sum of two
independent random variables $\geom(p_i) + \unif([0,1])$.
Given a $t$, the defender will return within at most $t$ steps if and
only if $\geom(p_i) \leq \Floor{t}$ or 
$\geom(p_i) = 1+\Floor{t}$ and $\unif([0,1]) \leq (t \mod 1)$.
The two events are disjoint, the first one having probability
$1-(1-p_i)^{\Floor{t}}$, and the second having probability
$p_i \cdot (1-p_i)^{\Floor{t}} \cdot (t \mod 1)$. 
Hence, 
$\Cdf[i]{t} = 1 - (1 - p_i \cdot (t \mod 1)) \cdot (1-p_i)^{\Floor{t}}$,
and the attacker's utility from attacking target $i$ for $t$ time
units is 
\[
\Value{i} \cdot t \cdot (1-\Cdf[i]{t}) 
\; = \; \Value{i} \cdot t \cdot (1 - p_i \cdot (t \mod 1)) \cdot (1-p_i)^{\Floor{t}}.
\]
Writing $t = x + k$ for an integer $k = \Floor{t}$ and $x = (t \mod 1)
\in [0,1)$, a derivative test shows that the expression is monotone
decreasing in $x$ for any $k \geq 1$, whereas for $k = 0$, it has a
local maximum at $x = \frac{1}{2p_i} \geq 1$.
Because the latter is not feasible, we only need to consider the case
$(t \mod 1) = 0$ for the remainder, so the attacker's utility simplifies
to $\Value{i} \cdot t \cdot (1-p_i)^{\Floor{t}}$.

Taking a derivative with respect to $t$ and setting it to 0 gives us
that the unique local extremum is at $t = \frac{-1}{\ln(1-p_i)}$,
where the attacker's utility is $\frac{p_i}{\e \cdot \ln(1/(1-p_i))}$.
This local extremum is a maximum because the attacker's utility at $t=0$
and $t=\infty$ is 0. 

A derivative test and Taylor series bound shows that 
$\frac{p_i}{\e \cdot \ln(1/(1-p_i))}$
is monotone decreasing in $p_i$, so it is maximized as $p_i \to 0$,
where it converges to $1/\e$.
Notice that as $p_i \to 0$, there are infinitely many values of $p_i$
for which $\frac{-1}{\ln(1-p_i)}$ is an integer, so the choice of $t$
in our previous optimization is indeed valid.

Under an optimal schedule, the attacker's expected utility is
\quarter, completing the proof of the approximation guarantee.
\end{proof}

\section{Formalization of Notions about Schedules}
\label{sec:appendix-formalization}

\subsection{Canonical Schedules}
The general definition of defender schedules allows for strange
schedules that are clearly suboptimal, such as the defender leaving a
target $i$ and returning to it shortly afterwards, or visiting a
target infinitely often within a bounded time interval with shorter
and shorter return times. 
For ease of notation and analysis, we would like to rule out such
schedules. 
The following definition captures ``reasonable'' schedules.

\begin{definition}[canonical schedules] \label{def:canonical}
We say that a valid schedule \LOC is \emph{canonical}
if $\R^+$ can be partitioned into countably many disjoint intervals
$I_1, I_2, I_3, \ldots$ with the following properties:
\begin{enumerate}
\item All odd intervals $I_{2k-1}$ are open and of length exactly 1, 
  and $\Loc{t} = \perp$ if and only if $t \in \bigcup_k I_{2k-1}$.
\item All even intervals $I_{2k}$ are closed. 
(Even intervals could consist of a single point.) 
\end{enumerate}

A defender mixed schedule \LOCDIST is canonical if it is a distribution
over canonical deterministic schedules.
\end{definition}
Note that it follows from validity that any
canonical \LOC is constant on the even intervals.

Intuitively, a canonical schedule is one in which the defender travels
as quickly as possible (in one unit of time) from one target to the
next target, visits it for some (possibly zero) time, then travels to
the next (necessarily different) target, etc. 
That we may focus on canonical schedules w.l.o.g.~is captured by the
following proposition: 

\begin{proposition}
  \label{lem:canonical}
  For each valid schedule \LOC, there exists a canonical
  schedule $\LOC'$ that is at least as good for the defender, in the
  sense that for any choice \Attack{i}{\AStart}{\ALength} of the attacker,
\begin{align*}
  \AUtilFull{\LOC'}{\Attack{i}{\AStart}{\ALength}} \leq \AUtilFull{\LOC}{\Attack{i}{\AStart}{\ALength}}.
\end{align*}
\end{proposition}

\begin{emptyproof}
  Given \LOC, define $\LOC'$ as follows.
  \begin{enumerate}
  \item For every $t$ with $\Loc{t} \neq \perp$ let
    $\LOC'(t)=\LOC(t)$.
  \item For every $t$ with $\LOC(t) = \perp$ 
    \begin{enumerate}
    \item If $t$ is in the closure of $\LOC^{-1}(i)$, set
      $\LOC'(t) = i$.
    \item Denote by $i(t)$ the last target visited before time $t$
      (setting $i(t)=1$ if none exists) and by $j(t)$ the first target visited
      after time $t$ (again setting $j(t)=1$ if none exists). Note that
      $i(t)$ and $j(t)$ are well-defined because \LOC is valid; this
      would not in general be true for an arbitrary
      $\LOC \colon \R^+ \to \{1,\ldots,n,\perp\}$.
    \item If $i(t)=j(t)$ then set $\LOC'(t) = i(t)$. That is, if in
      \LOC, the defender leaves a target $i$ and then comes back to it
      without visiting another, then in $\LOC'$, the defender just
      stays at $i$.
      (In addition to the defining property of being canonical,
      this ensures that no target $i$ is visited twice in a row.)
    \item If $i(t) \neq j(t)$ and the difference between $t$ and
      $\inf\Set{\tau > t}{\Loc{\tau} = j(t)}$ is at least 1, 
      then set $\LOC'(t)=i(t)$. That is, if the defender took more than one
      unit of time to reach target $j(t)$ from $i(t)$, then she might
      as well have stayed at $i(t)$ until one time unit before getting
      to $j(t)$.
    \item Otherwise, set $\LOC'(t) = \perp$.
    \end{enumerate}
  \end{enumerate}
It is easy to verify that $\LOC'$ is indeed canonical.
Consider any choice of attack \Attack{i}{\AStart}{\ALength}.
Because the preceding transformations only replaced $\perp$ (i.e.,
transit) times with times at targets,
whenever the attacker is not caught in $\LOC'$,
he was not caught in $\LOC$, so his utility can only decrease:
\begin{align*}
  \AUtilFull{\LOC'}{\Attack{i}{\AStart}{\ALength}}
  \leq \AUtilFull{\LOC}{\Attack{i}{\AStart}{\ALength}}.\QED
\end{align*}
\end{emptyproof}

\subsection{Shift Invariance}
To simplify the analysis, we would like to restrict our attention to
\emph{shift invariant} schedules for the defender: 
schedules such that the attacker's and defender's utilities depend
only on the duration $t'-t$ of the attack, but not on the start time
$t$. We formally define this notion as follows, and show that this
restriction is without loss of generality, as there is always an
optimal shift-invariant schedule.
For each $\tau \in \R^+$, define the \emph{shift operator}
$\SHIFTOP{\tau}: \STRAT \to \STRAT$ by
\begin{align*}
  [\ShiftOp{\tau}{\LOC}](t) = \Loc{t+\tau}.
\end{align*}
That is, the pure schedule \ShiftOp{\tau}{\LOC} is equal to \LOC, 
but leaves out the first $\tau$ time units of \LOC, 
shifting the remainder of the schedule forward in time.
Note that
\begin{align}
\AUtilFull{\ShiftOp{\tau}{\LOC}}{\Attack{i}{\AStart}{\ALength}}
& = \AUtilFull{\LOC}{\Attack{i}{\AStart+\tau}{\ALength}}.
\label{eqn:shift-definition}
\end{align}

The operator \SHIFTOP{\tau} extends naturally to act on mixed schedules
\LOCDIST.\footnote{A measurable map $P \colon X \to X$ can be extended
  to a linear operator on probability measures on $X$ as follows:
  For any measurable subset $A \subseteq X$, define
  $[P(\mu)](A) = \mu(P^{-1}(A))$. This defines a mapping 
  $\mu \mapsto P(\mu)$.}
We say that a mixed schedule
\LOCDIST is \emph{shift-invariant} if
$\ShiftOp{\tau}{\LOCDIST} = \LOCDIST$ for all $\tau \in \R^+$.  
The following lemma 
shows that an optimal schedule for the defender exists,
and that we may focus on shift-invariant schedules without loss of
generality.

\begin{lemma} \label{lem:shift-invariant}
The defender has an optimal mixed schedule that is shift-invariant.
\end{lemma}

\begin{proof}
  To prove this lemma, we introduce a natural topology on \STRAT, the
  space of valid canonical pure strategies. This topology is related
  to the Skorohod topology~\cite{skorokhod1957limit}.
  Given a $\LOC \in \STRAT$, define
  $\CLOC \colon \R^+ \to \{1,\ldots,n\}$ by setting $\CLOC(t)$
  to be either \Loc{t}, if $\Loc{t} \neq \perp$,
  or else setting it to be the first target visited after time $t$.
  Thus, $\CLOC(t)$ is the   target visited at time $t$,
  or the target that the defender is en route to visiting.
  Note that $\CLOC^{-1}(i)$ is the union of a
  countable set of intervals of length at least $1$, each open on the
  left and closed on the right.  Note also that the map
  $\LOC \mapsto \CLOC$ is ``almost'' invertible; since travel times
  are always 1, we know when each visit to each target began.  The
  exception is the first visit, and so \LOC is determined by \CLOC,
  together with the time of the beginning of the first target visit,
  which is always at most 1.

  The topology on \STRAT is the topology of convergence in
  $\mathcal{L}^1$ on compact sets.
    Specifically, for any $t_1, t_2 \in \R^+$, define
  $\Delta_{t_1,t_2}(\LOC',\LOC)$ to be the measure of the subset of
  $[t_1,t_2]$ on which at least one of the following two holds:
  (1) $\LOC' \neq \LOC$, or (2) $\CLOC' \neq \CLOC$.
  Then, we say that the limit of $\LOC_m$ for $m \to \infty$
  is equal to \LOC iff $\Delta_{t_1,t_2}(\LOC_{m},\LOC) \to 0$ 
  for all $t_1,t_2 \in \R^+$.
  It is straightforward to verify that this
  topology is compact and metrizable.\footnote{The metric is
    $\sum_{m=1}^\infty 2^{-m}d_{m}(\LOC_1,\LOC_2)$, where
    $d_{m}(\LOC_1,\LOC_2)$ is the measure of the subset of $[0,m]$ in
    which either $\LOC_1$ and $\LOC_2$ differ, or $\CLOC_1$ and
    $\CLOC_2$ differ.}  Hence the corresponding weak* topology on
  mixed strategies is also compact. Note also that the shift operator
  $\SHIFTOP{\tau} \colon S \to S$ is continuous in this topology.
  
  Note that if $\LOC_{m} \to_{m \to \infty} \LOC$, 
  and if target $i$ is visited in $[\AStart,\AStart+\ALength]$ in every
  $\LOC_{m}$, 
  then it is also visited in $[\AStart,\AStart+\ALength]$ in \LOC.
  Hence,
  \begin{align*}
    \lim_{m \to \infty} \AUtilFull{\LOC_{m}}{\Attack{i}{\AStart}{\ALength}}
    \geq \AUtilFull{\LOC}{\Attack{i}{\AStart}{\ALength}},
  \end{align*}
  and so $\AUtilFull{\cdot}{\Attack{i}{\AStart}{\ALength}}$
  is a lower semi-continuous map from $\STRAT$ to $\R^+$.
  It follows that 
  \begin{align*}
    \LOCDIST \mapsto \Expect[\LOC \sim \LOCDIST]{\AUtilFull{\LOC}{\Attack{i}{\AStart}{\ALength}}}
  \end{align*}
  is lower semi-continuous as well. Hence
  \begin{align*}
    \AUtil{\LOCDIST} = \sup_{i,\AStart,\ALength}\Expect[\LOC \sim \LOCDIST]{\AUtilFull{\LOC}{\Attack{i}{\AStart}{\ALength}}}
  \end{align*}
  is also lower-semicontinuous, and thus attains a minimum on the
  compact space of mixed strategies. Thus we have shown that an
  optimal schedule exists.
  
  When the attacker can obtain expected utility $u$ against
  \ShiftOp{\tau}{\LOC} by choosing \Attack{i}{\AStart}{\ALength},
  he can obtain the same utility $u$ against \LOCDIST
  by choosing \Attack{i}{\AStart+\tau}{\ALength}. 
  Therefore, the defender's utility is (weakly) monotone in $\tau$, 
  in the following sense:
  \begin{align}
    \AUtil{\ShiftOp{\tau}{\LOCDIST}} & \leq \AUtil{\LOCDIST}.
                                       \label{eq:monotone}
  \end{align}
  
  Let \LOCDIST[1] and \LOCDIST[2] be mixed strategies, 
  and let $\LOCDIST = \beta \LOCDIST[1] + (1-\beta) \LOCDIST[2]$
  be the schedule in which \LOCDIST[1] is carried out with probability
  $\beta$ and \LOCDIST[2] with probability $1-\beta$. 
  Since suprema are subadditive, the attacker's utility is convex:
  \begin{align}
    \AUtil{\LOCDIST} & \leq \beta \AUtil{\LOCDIST[1]} 
                       + (1-\beta)\AUtil{\LOCDIST[2]}.
                       \label{eq:concave}
  \end{align}
  
  Let \LOCDIST be an optimal mixed schedule. 
  For $m \in \N$ let
  \begin{align*}
  \LOCDIST[m] = \frac{1}{m}\int_0^{m} \ShiftOp{\tau}{\LOCDIST}\,\dd\tau.
  \end{align*}
  By the monotonicity (Eq.~\eqref{eq:monotone})
  and convexity (Eq.~\eqref{eq:concave}) of \AUtil{\LOCDIST},
  we have that $\AUtil{\LOCDIST[m]} \leq \AUtil{\LOCDIST}$.
  
  Since \STRAT is compact, the sequence $(\LOCDIST[m])_m$ has a
  converging subsequence that converges to some \LOCDIST[\infty]. 
  By the lower semi-continuity of \AUtil{\LOCDIST}, 
  \begin{align*}
    \AUtil{\LOCDIST[\infty]}
    & \leq \lim_{m \to \infty} \AUtil{\LOCDIST[m]} 
    \; \leq \; \AUtil{\LOCDIST};
  \end{align*}
  therefore \LOCDIST[\infty] is also optimal. 
  Finally, \LOCDIST[\infty] is by construction shift-invariant.
\end{proof}

\subsection{Transitive and Ergodic Schedules}
\label{sec:tran-erg}

We say that a shift-invariant mixed schedule \LOCDIST is
\emph{transitive} if almost every pure
schedule $\LOC_0$ chosen from \LOCDIST is periodic with some
period $\tau$ (i.e.,
$\ShiftOp{\tau}{\LOC_0}=\LOC_0$) and
\begin{align*}
  \LOCDIST = \frac{1}{\tau}\int_0^{\tau}\delta_{\ShiftOp{t}{\LOC_0}}\,\dd t,
\end{align*}
where $\delta_{\LOC}$ is the point mass on \LOC. Intuitively,
\LOCDIST simply repeats the same periodic schedule, with a phase
chosen uniformly at random.

A weaker property of a shift-invariant mixed schedule \LOCDIST is
\emph{ergodicity}: \LOCDIST is ergodic if almost every
pure schedule $\LOC_0$ chosen from \LOCDIST satisfies
\begin{align*}
  \LOCDIST = \lim_{\tau \to \infty}\frac{1}{\tau}\int_0^{\tau}\delta_{\ShiftOp{t}{\LOC_0}}\,\dd t.
\end{align*}
In fact, this is not the usual definition of an ergodic measure, but
the conclusion of the Ergodic Theorem. An equivalent property is that
$\LOCDIST$ cannot be written as the convex combination $\LOCDIST =
\beta \LOCDIST_1 + (1-\beta)\LOCDIST_2$ of two different 
shift-invariant measures. That is, $\LOCDIST$ is an extremal point
in the convex set (simplex, in fact) of shift-invariant measures.

\subsection{Times Between Visits to Targets}

We now more formally define the notion of the (random) time
between visits to a target $i$. While the notion is intuitively clear,
for arbitrary defender strategies \LOCDIST, a precise definition
requires some subtlety. We give a general definition for arbitrary
mixed schedules, not just random sequences.

Having defined schedules on $\R^+$, we now extend the
definition to schedules on $[-\tau, \infty)$ and eventually to $\R$,
using a standard construction called the \emph{bi-infinite extension}.
We define a modified shift operator $\ShiftOpT{\tau}{\cdot}$,
mapping schedules ($\LOC \colon \R^+ \to \SET{1,\ldots,n,\perp}$)
to \emph{$\tau$-schedules} 
$\LOC' \colon [-\tau, \infty) \to \SET{1,\ldots,n,\perp}$, via
$[\ShiftOpT{\tau}{\LOC}](t) = \LOC(t+\tau)$.
Thus, $\ShiftOpT{\tau}{\LOC}$ is simply a version of \LOC shifted
$\tau$ units to the left.
The map $\ShiftOpT{\tau}{\cdot}$ extends to a map on mixed schedules in the
obvious way.
For any shift-invariant mixed schedule \LOCDIST,
$\ShiftOpT{\tau}{\LOCDIST}$ is also shift-invariant, and furthermore,
for any $\tau' < \tau$, the distribution $\ShiftOpT{\tau}{\LOCDIST}$,
projected to $[-\tau',\infty)$, is the same distribution as
$\ShiftOpT{\tau'}{\LOCDIST}$. It follows that

\begin{align}
  \label{eq:bi-infinite}
  \LOCDIST_\infty = \lim_{\tau \to \infty}\ShiftOpT{\tau}{\LOCDIST}
\end{align}
is a well defined measure on pure schedules that are functions
$\LOC_\infty \colon \R \to \{1,\ldots,n,\perp\}$. We call
$\LOCDIST_\infty$ the \emph{bi-infinite extension} of \LOCDIST. 
It is straightforward to verify that it, too, is shift-invariant. 
Note that the distribution of the first visit to $i$ at non-negative
times, $\RTIME[i] = \min \Set{t \geq 0}{\Loc{t} = i}$, 
has the same distribution under $\LOCDIST_\infty$ as under \LOCDIST,
since the restriction of $\LOCDIST_\infty$ to non-negative times is
equal to \LOCDIST.

Given a target $i$ and a shift-invariant mixed schedule \LOCDIST,
let $\tilde \LOC \colon (-\infty, \infty) \to \{1,\ldots,n,\perp\}$
be a random schedule with distribution $\LOCDIST_\infty$.  
Let \TBV[i] be the (random) time between the last visit to $i$ before
time zero, until the first visit to $i$ after time zero:
\begin{align*}
  \TBV[i] = (\inf \Set{t \geq 0}{\tilde\LOC(t) = i}) 
         - (\sup \Set{t \leq 0}{\tilde\LOC(t) = i}).
\end{align*}
The choice of time $0$ here is immaterial because of shift
invariance. 
\TBV[i] could be infinite, but this will never happen in an optimal
\LOCDIST, because it would imply that the attacker's expected utility
for choosing $i$ is infinite; 
we hence assume henceforth that $\Prob{\TBV[i] = \infty} = 0$. 
Finally, contrary to what one might intuitively guess, 
even for transitive \LOCDIST, the distribution of \TBV[i] is \emph{not}
the same as the long-run empirical distribution of times between
visits, as gaps are chosen at time $0$ with probability proportional
to their length. 
The same holds for general \LOCDIST.

\section{Tightness of the $2$-Quasi-Regularity Result}
\label{app:tightness}

In this section, we prove Proposition~\ref{thm:2-quasi-tight}.
  For convenience, we restate the proposition here:

\begin{rtheorem}{Proposition}{\ref{thm:2-quasi-tight}}
  Let $n=3$ and $\FreqVec=(1/2,1/3,1/6)$.
  Then, for every $\epsilon>0$,
  there are no $(2-\epsilon)$-quasi-regular random sequences.
\end{rtheorem}

\begin{proof}
  Let \SEQ be a $(2-\epsilon)$-quasi-regular random sequence. 
  We claim that $\TBV[1] = 2$ with probability 1, and
  $\TBV[2] = 3$ with probability 1.
  For suppose that with positive probability $\TBV[1] \leq 1$.
  Then, because $\EL[1] = 2$, we also would have to have
  $\TBV[1] \geq 3$ with positive probability, and vice versa.
  Similarly, $\TBV[2] \leq 2$ with positive probability iff 
  $\TBV[2] \geq 4$ with positive probability.
  Either of those cases would lead to a ratio ($3/1$ or $4/2$) larger
  than $2-\epsilon$, violating $(2-\epsilon)$-quasi-regularity.

  Hence, with probability one, target $1$ appears in every other time
  period and target $2$ appears in every third time period, which is
  impossible.
\end{proof}

\section{Proof of Theorem~\ref{THM:THREE-FIBONACCI}}
\label{sec:slater-proof}

In this section, we prove Theorem~\ref{THM:THREE-FIBONACCI},
  restated here for convenience:

\begin{rtheorem}[Slater \cite{slater1950gaps}]{Theorem}{\ref{THM:THREE-FIBONACCI}}
Assume that $\FFR \leq \half$.
Let $k$ be smallest such that 
\begin{align*}
|\FIB{k+1}/\GR - \FIB{k}| \; \leq \; \FFR.
\hfill 
\end{align*}

Then, the distribution of return times is
\begin{equation*}
  \begin{split}
    \Prob{\TBV[i] = \FIB{k+1}} & = \FIB{k+1} \cdot \left(\FFR -
      (1/\GR)^{k+1} \right),\\
    \Prob{\TBV[i] = \FIB{k+2}} & = \FIB{k+2} \cdot \left(\FFR -
      (1/\GR)^{k+2} \right),\\
    \Prob{\TBV[i] = \FIB{k+3}} & = \FIB{k+3} \cdot \left(-\FFR +
      (1/\GR)^{k} \right),\\
    \Prob{\TBV[i] = t} & = 0 \mbox{ for all other } t.
\end{split}
\end{equation*}
\end{rtheorem}

We begin with a few simple, but useful, technical lemmas.
First, we give a closed form for expressions of the form
$\FIB{k+1} - \GR \FIB{k}$.

\begin{lemma} \label{lem:fibonacci-difference}
For any $k$, we have that
$
\FIB{k+1} - \GR \FIB{k}
\; = \; (-1/\GR)^k.
$
\end{lemma}

\begin{emptyproof}
Using the closed-form expression for Fibonacci Numbers 
(Part~\ref{sqrt-five} of Lemma~\ref{lem:fibonacci-basic}),
we can write

\begin{align*}
\FIB{k+1} - \GR \FIB{k} 
& =  \frac{  \left( \GR^{k+1} - (-1/\GR)^{k+1} \right)
          - \left( \GR^{k+1} - \GR (-1/\GR)^k \right)}{\sqrt{5}}\\
& =  \frac{-(-1/\GR)^{k+1} + \GR (-1/\GR)^{k}}{\sqrt{5}}\\
& =  (-1/\GR)^k \cdot \frac{1/\GR + \GR}{\sqrt{5}}
\; = \; (-1/\GR)^k \GR. \QED
\end{align*}
\end{emptyproof}

\begin{lemma} \label{lem:circle-fibonacci}
\begin{enumerate}
\item \label{lem:circle-fibonacci-equivalence}
For all $\delta \in (-\half,\half]$ and integers 
$N \geq 1$, the following two are equivalent:
\begin{itemize}
\item $\delta = (\GR N) \mod 1$.
\item There exists a positive integer $D$ with
$N/D - \GR = \delta \GR/D$.
\end{itemize}

\item \label{lem:circle-fibonacci-unique}
Let $\delta = (\GR \FIB{k}) \mod 1$ for $k \geq 2$
(where we consider the range of the $\mod$ operation to be $(-\half, \half]$).
Then, $\delta = \FIB{k}/\GR - \FIB{k-1}$.
\end{enumerate}
\end{lemma}

\begin{emptyproof}
\begin{enumerate}
\item Because $\GR = 1 + 1/\GR$, 
the first condition can be rewritten as 
$(1+1/\GR) N = \delta + D'$ for some integer $D'$.
When $\delta < 0$, we must have $D' > N$;
when $\delta \geq 0$, because $N/\GR > \half \geq \delta$,
we again have $D' > N$.
The preceding equality can therefore be rearranged to
$N/\GR = \delta + (D'-N)$. 
Multiplying by $\GR/(D'-N)$ now gives equivalence with the second
condition, writing $D = D'-N$.
\item In the first part of the lemma, set $N=\FIB{k}$. 
Then, the condition is equivalent to the existence of a positive
integer $D$ with $\FIB{k}/D - \GR = \delta \GR/D$,
implying that $|\FIB{k} - D \GR| = |\delta| \GR$.
By choosing $D = \FIB{k-1}$,
according to Lemma~\ref{lem:fibonacci-difference}, we get that
\[
|\FIB{k}/\GR - \FIB{k-1}| 
\; = \; (1/\GR)^k
\; \stackrel{k \geq 2}{\leq} \; \half.
\]
Therefore, for any $D \neq \FIB{k-1}$,
we get that $|\FIB{k}/\GR - D| > 1-\half = \half$,
meaning that no $D \neq \FIB{k-1}$ can satisfy
$\FIB{k}/D - \GR = \delta \GR/D$.
By substituting the unique choice $D = \FIB{k-1}$,
we obtain the second part of the lemma.\QED
\end{enumerate}
\end{emptyproof}

Because the Fibonacci numbers are the convergents of the Golden Ratio,
they provide the best rational approximation, in the following sense.

\begin{theorem}
\label{thm:fibonacci-approximation}
Let $\hat{M} \geq 1$ be arbitrary.
Let $k$ be the largest even number with $\FIB{k} \leq \hat{M}$, 
and $k'$ the largest odd number with $\FIB{k'} \leq \hat{M}$.
Then, for all $M \leq \hat{M}$ and all $N$, we have the following:
\begin{enumerate}
\item $N/M > \GR$ implies $\FIB{k+1}/\FIB{k} \leq N/M$.
\label{thm:fibonacci-approximation:greater}
\item $N/M < \GR$ implies $\FIB{k'+1}/\FIB{k'} \geq N/M$.
\label{thm:fibonacci-approximation:smaller}
\end{enumerate}
\end{theorem}

Theorem~\ref{thm:fibonacci-approximation} follows directly from
standard results stating that the convergents provide the best
approximation to real numbers (e.g., \cite[p.~11]{lovasz},
noting that the second (intermediate) case cannot happen for the
Golden Ratio).

We are now ready to prove the characterization of the distribution of
\TBV[i] under the Golden Ratio schedule.

\begin{extraproof}{Theorem~\ref{THM:THREE-FIBONACCI}}
We begin by showing that the support of return times consists only of
Fibonacci Numbers.
Consider the interval $I = [0,\FFR)$.
Let $m \geq 1$ be a return time. 
Let $x \in I$ be arbitrary, 
and $y = (x + \GR m) \mod 1$, which is in $I$ by assumption.
Define $\delta = y-x$.
Because both $x,y \in I$, we have that $\delta \in [-x, \FFR-x)$.
By Part~\ref{lem:circle-fibonacci-equivalence} of
Lemma~\ref{lem:circle-fibonacci},
there is a positive integer $D$ such that
\begin{align*}
m/D - \GR & = \delta\GR/D \in [-x\GR/D, (\FFR-x)\GR/D).
\end{align*}

We now distinguish two cases:
\begin{itemize}
\item If $\delta > 0$, then $m/D > \GR$,
so Case~\ref{thm:fibonacci-approximation:greater} of
Theorem~\ref{thm:fibonacci-approximation} implies
that the largest even $j$ such that $\FIB{j} \leq D$ satisfies
$\FIB{j+1}/\FIB{j} > \GR$ and $\FIB{j+1}/\FIB{j} \leq m/D$.
Thus, $(x + \GR \FIB{j+1}) \mod 1 \in I$,
meaning that the defender returns to the target in $\FIB{j+1}$ steps.
Because $D \geq \FIB{j}$ and $m/D \geq \FIB{j+1}/\FIB{j}$,
we get that $m \geq \FIB{j+1}$;
unless $m = \FIB{j+1}$, this would contradict the definition of $m$ as
a return time, so we have shown that $m$ is a Fibonacci number.

\item Similarly, if $\delta < 0$, and thus $m/D - \GR < 0$,
then Case~\ref{thm:fibonacci-approximation:smaller} of
Theorem~\ref{thm:fibonacci-approximation} implies
that the largest odd $j$ such that $\FIB{j} \leq D$ satisfies
$\FIB{j+1}/\FIB{j} - \GR < 0$ and $\FIB{j+1}/\FIB{j} \geq m/D$.
By the same argument, we obtain now that $m=\FIB{j+1}$.
\end{itemize}

\medskip

Next, we prove the second part of the theorem.
First, notice that the $k$ defined in the theorem actually exists.
By Lemma~\ref{lem:fibonacci-difference}, we get that
$|\FIB{k+1} - \GR \FIB{k}| = (1/\GR)^k \to 0$ as $k \to \infty$,
so there exists a $k$ (and thus a smallest $k$) with
$|\FIB{k+1} - \GR \FIB{k}| \leq \GR \FFR$.

We show that there cannot be a return time $m < \FIB{k+1}$.
If there were, then by the previous part of the proof,
$m$ would be a Fibonacci number, say, $m = \FIB{\ell}$.
And because $m \geq 1$, we get that $\ell \geq 2$.
By Part~\ref{lem:circle-fibonacci-unique} of
Lemma~\ref{lem:circle-fibonacci}, that means that 
$\FIB{\ell}/\GR - \FIB{\ell-1} = y-x$,
and hence
$|\FIB{\ell}/\GR - \FIB{\ell-1}| = |y-x| < \FFR$,
contradicting the definition of $k$ as smallest with that property.

Consider a return to $I$ within $m$ steps,
starting from $x \in I$ and ending at $y \in I$,
so that $\delta_{\ell} := y-x$ satisfies
$|\delta_{\ell}| < \FFR$.
By the preceding analysis,
$m = \FIB{\ell}$ for some $\ell \geq k+1$.
Again, by Part~\ref{lem:circle-fibonacci-unique} of
Lemma~\ref{lem:circle-fibonacci}, we obtain that
$\delta_{\ell} = \FIB{\ell}/\GR - \FIB{\ell-1}$.

When $\delta_{\ell} < 0$, the $x \in I$ with 
$x+\delta_{\ell} \in I$ are exactly captured by the interval
$J_{\ell} := [|\delta_{\ell}|,\FFR]$, 
while for $\delta_{\ell} > 0$, they are exactly the interval
$J_{\ell} := [0,\FFR-\delta_{\ell})$.
In either case, the interval $J_{\ell}$ has size exactly
$|J_{\ell}| = \FFR - |\delta_{\ell}|$.

We will show that $J_{k+2} \cup J_{k+3} = I$.
By Lemma~\ref{lem:fibonacci-basic}, the signs of $\delta_{\ell}$ are
alternating, meaning that the intervals $J_{\ell}$ alternate being of
the form $[0,x]$ and $[y,\FFR)$. 
In particular, to show that $J_{k+2} \cup J_{k+3} = I$, it suffices to
show that $|J_{k+2}| + |J_{k+3}| \geq \FFR$.
Because
$|J_{k+2}| + |J_{k+3}| = 2\FFR - |\delta_{k+2}| - |\delta_{k+3}|$,
this is equivalent to showing that 
$|\delta_{k+2}| + |\delta_{k+3}| \leq \FFR$.
We distinguish two cases, based on whether $k$ is even or odd.
\begin{itemize}
\item If $k$ is even, then
$\delta_{k+2} = \FIB{k+2}/\GR - \FIB{k+1} < 0$
and $\delta_{k+3} = \FIB{k+3}/\GR - \FIB{k+2} > 0$, so we obtain that
\[
|\delta_{k+3}| + |\delta_{k+2}|
\; = \; \FIB{k+3}/\GR - \FIB{k+2} - \FIB{k+2}/\GR + \FIB{k+1}
\; = \; \FIB{k+1}/\GR - \FIB{k}
\; = \; |\FIB{k+1}/\GR - \FIB{k}|
\; \leq \; \FFR,
\]
by the definition of $k$.

\item If $k$ is odd, then
$\delta_{k+2} = \FIB{k+2}/\GR - \FIB{k+1} > 0$
and $\delta_{k+3} = \FIB{k+3}/\GR - \FIB{k+2} < 0$, so we obtain that
\[
|\delta_{k+2}| + |\delta_{k+3}|
\; = \; \FIB{k+2}/\GR - \FIB{k+1} - \FIB{k+3}/\GR + \FIB{k+2}
\; = \; \FIB{k} - \FIB{k+1}/\GR
\; = \; |\FIB{k+1}/\GR - \FIB{k}|
\; \leq \; \FFR.
\]
\end{itemize}

Thus, we have shown that the support of the distribution is indeed
contained in $\SET{\FIB{k+1}, \FIB{k+2}, \FIB{k+3}}$.
Finally, we can work out the frequencies.
Conditioned on being in the interval of size \FFR,
the probability of being in $J_{\ell}$ is 
$q_{\ell} = |J_{\ell}|/\FFR$.
To arrive at the attacker's observed distribution of \TBV[i],
we notice that the probability of time 0 being in an interval of
length $\FIB{\ell}$ is 
\[
\frac{q_{\ell} \FIB{\ell}}{\sum_j q_j \FIB{j}}
= \frac{q_{\ell} \FIB{\ell}}{1/\FFR}
= \FIB{\ell} \cdot |J_{\ell}|.
\]
Thus, we obtain that

\begin{align*}
\Prob{\TBV[i] = \FIB{k+1}} & = \FIB{k+1} \cdot |J_{k+1}| 
\; = \; \FIB{k+1} \cdot (\FFR - |\FIB{k+1}/\GR - \FIB{k}|),\\
\Prob{\TBV[i] = \FIB{k+2}} & = \FIB{k+2} \cdot |J_{k+2}| 
\; = \; \FIB{k+2} \cdot (\FFR - |\FIB{k+1} - \FIB{k+2}/\GR|),\\
\Prob{\TBV[i] = \FIB{k+3}} & = 1 - q_1 - q_2 \\
& = \FIB{k+3} \cdot (-\FFR + |\FIB{k+1}/\GR - \FIB{k} + \FIB{k+1} - \FIB{k+2}/\GR|)\\
& = \FIB{k+3} \cdot (-\FFR + |\FIB{k-1} - \FIB{k}/\GR|).\\
\end{align*}

Notice that we arranged the terms inside absolute values such that for
even $k$, they are all positive, while for odd $k$, they are all negative. 
This allowed us to simply add inside the absolute value.
Applying Lemma~\ref{lem:fibonacci-difference} to all three terms
now completes the proof.
\end{extraproof}

\section{Computational Considerations for the Golden Ratio Schedule}
\label{sec:finite-precision}

As phrased, Algorithm~\ref{alg:golden-ratio} requires precise
arithmetic on irrational numbers, and drawing a uniformly random
number from $[0,1]$. Here, we discuss how to implement the algorithm
such that each target $i$ visited in step $t$ can be computed in time
polynomial in the input size.

Let $\Freq{i} = a_i/b_i$ for each $i$, and write
$M = \lcm(b_1, \ldots, b_m)$ for the common denominator.
Notice that $\log M \leq \sum_i \log b_i$ is polynomial in the
input size.

For each $i$, the number $P_i = \sum_{i' < i} p_{i'}$ is rational.
To decide whether target $i$ is visited in step $t$,
the algorithm needs to decide if 
$(\OFFSET + t/\GR) \mod 1 \in [P_i, P_{i+1}]$,
or --- equivalently --- if there is an integer $D$ with
$\OFFSET + t/\GR \in [D+P_i, D+P_{i+1}]$.
To decide whether $\OFFSET + t/\GR < D+P_j$ or $\OFFSET + t/\GR > D+P_j$ 
(for $j \in \SET{i,i+1}$), the algorithm needs to decide if
$\GR < \frac{t}{D+P_j-\OFFSET}$ or $\GR > \frac{t}{D+P_j-\OFFSET}$.
The key question is how many digits of \GR the algorithm needs to
evaluate for this decision, and how many digits of the uniformly
random offset \OFFSET it needs to decide on.

Suppose that the algorithm has generated the first $k$ random digits
of \OFFSET, having committed to 
$\frac{\ell}{10^k} \leq \OFFSET < \frac{\ell+1}{10^k}$
for some $\ell \in \SET{0,1,\ldots,10^k-1}$. 
Writing $P_j = N_j/M$ (using the denominator $M$ defined above),
a decision about target $P_j$ can be made whenever 
$\GR < \frac{t M 10^k}{10^k \cdot M D + 10^k N_j- M \cdot \ell}$
or
$\GR > \frac{t M 10^k}{10^k \cdot M D + 10^k N_j- M \cdot (\ell+1)}$.
In both cases, the right-hand side is a rational approximation to \GR
with denominator bounded by $\hat{M} := 2 \cdot 10^k \cdot M D$.

It is well known 
(see, e.g., \cite[Theorems~193--194]{hardy:wright})
that $|\GR - \frac{\hat{N}}{\hat{M}}| \geq
\frac{1}{(\sqrt{5}-\epsilon) \hat{M}^2}$ for all $\epsilon > 0$.
In particular, this implies that evaluating \GR to 
within $O(\log \hat{M}^2) = O(k + \log M + \log D)$ digits
is sufficient to test whether 
$\GR < \frac{t M 10^k}{10^k \cdot M D + 10^k N_j- M \cdot \ell}$,
and whether 
$\GR > \frac{t M 10^k}{10^k \cdot M D + 10^k N_j- M \cdot (\ell+1)}$.
In either of these cases, the algorithm has resolved whether
$\GR < \frac{t}{D+P_j-\OFFSET}$.

The only case where the algorithm cannot resolve whether
$\GR < \frac{t}{D+P_j-\OFFSET}$ is when
\[
\frac{t M 10^k}{10^k \cdot M D + 10^k N_j- M \cdot \ell} 
\; < \; \GR \; < \; \frac{t M 10^k}{10^k \cdot M D + 10^k N_j- M \cdot (\ell+1)}.
\]
In this case, the number of digits for \OFFSET is insufficient.
Notice that there is a unique value of $\ell$ for which this happens,
so the probability of failure is at most $10^{-k}$.
Taking a union bound over all $n$ interval boundaries and all $t$
rounds, we see that in order to succeed with high probability, 
the number of digits of \OFFSET the algorithm needs to generate is
$O(\log n + \log t)$.

In particular, the computation and required randomness are polynomial.

\bibliographystyle{plain}
\bibliography{../bibliography/names,../bibliography/conferences,../bibliography/poaching}

\begin{thebibliography}{10}

\bibitem{angel-discrepancy}
Omer Angel, Alexander~E. Holroyd, James~B. Martin, and James Propp.
\newblock Discrete low-discrepancy sequences.
\newblock https://arxiv.org/abs/0910.1077, 2010.

\bibitem{avenhaus2005playing}
Rudolf Avenhaus and Morton~John Canty.
\newblock Playing for time: A sequential inspection game.
\newblock {\em European Journal of Operational Research}, 167(2):475--492,
  2005.

\bibitem{avenhaus2013distributing}
Rudolf Avenhaus and Thomas Krieger.
\newblock Distributing inspections in space and time--proposed solution of a
  difficult problem.
\newblock {\em European Journal of Operational Research}, 231(3):712--719,
  2013.

\bibitem{avenhaus2002inspection}
Rudolf Avenhaus, Bernhard von Stengel, and Shmuel Zamir.
\newblock Inspection games.
\newblock In Robert~J. Aumann and Sergiu Hart, editors, {\em Handbook of Game
  Theory, Vol. 3}, chapter~51, pages 1947--1987. North-Holland, Amsterdam,
  2002.

\bibitem{basilico2009leader}
Nicola Basilico, Nicola Gatti, and Francesco Amigoni.
\newblock Leader-follower strategies for robotic patrolling in environments
  with arbitrary topologies.
\newblock In {\em Proc. 8th Intl. Conf. on Autonomous Agents and Multiagent
  Systems}, pages 57--64, 2009.

\bibitem{chan:chin:pinwheel}
Mee~Yee Chan and Francis Y.~L. Chin.
\newblock Schedulers for larger classes of pinwheel instances.
\newblock {\em Algorithmica}, 9:425--462, 1993.

\bibitem{chekuri2010rounding}
Chandra Chekuri, Jan Vond\'{a}k, and Rico Zenklusen.
\newblock Dependent randomized rounding via exchange properties of
  combinatorial structures.
\newblock In {\em Proc. 51st IEEE Symp. on Foundations of Computer Science},
  pages 575--584, 2010.

\bibitem{chen2014panorama}
William Chen, Anand Srivastav, and Giancarlo Travaglini, editors.
\newblock {\em A Panorama of Discrepancy Theory}, volume 2107 of {\em Lecture
  Notes in Mathematics}.
\newblock Springer, 2014.

\bibitem{chung2016discrepancy}
Fan Chung and Ron Graham.
\newblock On the discrepancy of circular sequences of reals.
\newblock {\em Journal of Number Theory}, 164:52--65, 2016.

\bibitem{conitzer2006commit}
Vincent Conitzer and Tuomas Sandholm.
\newblock Computing the optimal strategy to commit to.
\newblock In {\em Proc. 7th ACM Conf. on Electronic Commerce}, pages 82--90,
  2006.

\bibitem{diamond1982minimax}
Harvey Diamond.
\newblock Minimax policies for unobservable inspections.
\newblock {\em Mathematics of Operations Research}, 7(1):139--153, 1982.

\bibitem{paws2017}
Fei Fang, Thanh~H. Nguyen, Rob Pickles, Wai~Y. Lam, Gopalasamy~R. Clements,
  Bo~An, Amandeep Singh, Brian~C. Schwedock, Milind Tambe, and Andrew Lemieux.
\newblock {PAWS} --- a deployed game-theoretic application to combat poaching.
\newblock {\em AI Magazine}, 2017.
\newblock (to appear).

\bibitem{fang2015green}
Fei Fang, Peter Stone, and Milind Tambe.
\newblock When security games go green: Designing defender strategies to
  prevent poaching and illegal fishing.
\newblock In {\em Proc. 24th Intl. Joint Conf. on Artificial Intelligence},
  pages 2589--2595, 2015.

\bibitem{fishburn:lagarias:pinwheel}
Peter~C. Fishburn and Jeffrey~C. Lagarias.
\newblock Pinwheel scheduling: achievable densities.
\newblock {\em Algorithmica}, 22:14--38, 2002.

\bibitem{gandhi2006rounding}
Rajiv Gandhi, Samir Khuller, Srinivasan Parthasarathy, and Aravind Srinivasan.
\newblock Dependent rounding and its applications to approximation algorithms.
\newblock {\em Journal of the ACM}, 53:324--360, 2006.

\bibitem{hardy:wright}
Godfrey~H. Hardy and Edward~M. Wright.
\newblock {\em An Introduction to the Theory of Numbers}.
\newblock Oxford University Press, fourth edition, 1960.

\bibitem{holroyd2010rotor}
Alexander~E. Holroyd and James Propp.
\newblock Rotor walks and markov chains.
\newblock {\em Algorithmic probability and combinatorics}, 520:105--126, 2010.

\bibitem{HMRTV:pinwheel-conference}
Robert Holte, Al~Mok, Louis Rosier, Igor Tulchinsky, and Donald Varvel.
\newblock The pinwheel: a real-time scheduling problem.
\newblock In {\em Proc. Twenty-Second Annual Hawaii International Conference on
  System Sciences. Vol.II: Software Track}, 1989.

\bibitem{holte:rosier:tulchinsky:varval:pinwheel}
Robert Holte, Louis Rosier, Igor Tulchinsky, and Donald Varvel.
\newblock Pinwheel scheduling with two distinct numbers.
\newblock {\em Theoretical Computer Science}, 100:105--135, 1992.

\bibitem{immorlica-kleinberg-bandits}
Nicole Immorlica and Robert~D. Kleinberg.
\newblock Recharging bandits.
\newblock Manuscript. Available upon request, 2017.

\bibitem{itai1984golden}
Alon Itai and Zvi Rosberg.
\newblock A golden ratio control policy for a multiple-access channel.
\newblock {\em IEEE Transactions on Automatic Control}, 29(8):712--718, 1984.

\bibitem{knuth1998art}
Donald~E. Knuth.
\newblock {\em The Art of Computer Programming: Sorting and Searching},
  volume~3.
\newblock Pearson Education, 1998.

\bibitem{KorzhykCP11}
Dmytro Korzhyk, Vincent Conitzer, and Ronald Parr.
\newblock Security games with multiple attacker resources.
\newblock In {\em Proc. 22nd Intl. Joint Conf. on Artificial Intelligence},
  pages 273--279, 2011.

\bibitem{lin:lin:pinwheel}
Shun-Shii Lin and Kwei-Jay Lin.
\newblock A pinwheel scheduler for three distinct numbers with a tight
  schedulability bound.
\newblock {\em Algorithmica}, 19:411--426, 1997.

\bibitem{lovasz}
L\'{a}szl\'{o} Lov\'{a}sz.
\newblock {\em An Algorithmic Theory of Numbers, Graphs and Convexity},
  volume~50 of {\em CBMS-NSF Regional Conference Series in Applied
  Mathematics}.
\newblock SIAM, 1986.

\bibitem{panwar1992golden}
Shivendra~S. Panwar, Thomas~K. Philips, and Mon~Song Chen.
\newblock Golden ratio scheduling for flow control with low buffer
  requirements.
\newblock {\em IEEE Transactions on Communications}, 40(4):765--772, 1992.

\bibitem{skorokhod1957limit}
Anatolii~Volodimirovich Skorokhod.
\newblock Limit theorems for stochastic processes with independent increments.
\newblock {\em Theory of Probability \& Its Applications}, 2(2):138--171, 1957.

\bibitem{slater1950gaps}
Noel~B. Slater.
\newblock The distribution of the integers $n$ for which $\{\theta n\} <
  \varphi$.
\newblock {\em Proceedings of the Cambridge Philosophical Society},
  46(4):525--534, 1950.

\bibitem{slater1967gaps}
Noel~B. Slater.
\newblock Gaps and steps for the sequence $n\theta \mod 1$.
\newblock {\em Mathematical Proceedings of the Cambridge Philosophical
  Society}, 63(04):1115--1123, 1967.

\bibitem{srinivasan2001level}
Aravind Srinivasan.
\newblock Distributions on level-sets with applications to approximation
  algorithms.
\newblock In {\em Proc. 42nd IEEE Symp. on Foundations of Computer Science},
  pages 588--597, 2001.

\bibitem{tambe2011security}
Milind Tambe.
\newblock {\em Security and Game Theory: Algorithms, Deployed Systems, Lessons
  Learned}.
\newblock Cambridge University Press, 2011.

\bibitem{tijdeman1973distribution}
Robert Tijdeman.
\newblock On a distribution problem in finite and countable sets.
\newblock {\em Journal of Combinatorial Theory, Series A}, 15(2):129--137,
  1973.

\bibitem{tijdeman1980chairman}
Robert Tijdeman.
\newblock The chairman assignment problem.
\newblock {\em Discrete Mathematics}, 32:323--330, 1980.

\bibitem{von1934marktform}
Heinrich von Stackelberg.
\newblock {\em Marktform und {G}leichgewicht}.
\newblock Springer, 1934.

\bibitem{vonstengel2016recursive}
Bernhard von Stengel.
\newblock Recursive inpection games.
\newblock {\em Mathematics of Operations Research}, 41:935--952, 2016.

\bibitem{vorobeychik2012adversarial}
Yevgeniy Vorobeychik, Bo~An, and Milind Tambe.
\newblock Adversarial patrolling games.
\newblock In {\em Proc. 11th Intl. Conf. on Autonomous Agents and Multiagent
  Systems}, pages 1307--1308, 2012.

\end{thebibliography}

\end{document}